\documentclass[preprint,12pt,times]{elsarticle}

\usepackage{amssymb}
\usepackage{amsmath}
\usepackage{booktabs}
\usepackage{multirow}
\usepackage{graphicx}
\usepackage{bm}
\usepackage[colorlinks=true, citecolor=blue]{hyperref}
\usepackage{float}
\usepackage{amsmath,amssymb,bm}
\usepackage[colorlinks=true,citecolor=blue]{hyperref}
\usepackage{graphicx}

\usepackage{subcaption}
\usepackage{caption} 
\usepackage{amsthm}

\newtheorem{theorem}{Theorem}[section]

\newtheorem{lemma}[theorem]{Lemma}
\newtheorem{proposition}[theorem]{Proposition}

\theoremstyle{definition}

\theoremstyle{remark}
\newtheorem{remark}[theorem]{Remark}
\floatstyle{ruled}
\newfloat{algorithm}{tbp}{loa}
\floatname{algorithm}{Algorithm}

\begin{document}
	
	\begin{frontmatter}
		
		\title{A finite-precision Lanczos-Golub-Welsch route to probability-table construction in resonance self-shielding} 
		
		\author{Beichen Zheng}
		
		\affiliation{organization={China Nuclear Data Center, China Institute of Atomic Energy}, Department and Organization
			city={Beijing},
			country={China}}
		
\begin{abstract}
This work reformulates Chiba's affine-order prescription as a polynomial-moment problem for a transformed positive measure, and develops an alternative finite-precision construction route based on this reformulation. The proposed construction proceeds through
discrete-measure realization, symmetric Lanczos reduction, and Golub--Welsch extraction, replacing the conventional moment--Pad\'e pipeline. The subgroup total levels and probabilities are obtained by a Gauss-type compression step that preserves nonnegative-realness, while the reaction-channel levels are recovered on the compressed nodes by orthogonal-basis matching. In five tested resonance-channel cases, the proposed construction
yields lower effective-cross-section errors and avoids the
order-induced emergence of complex responses
observed in the conventional construction.
\end{abstract}

		
	\end{frontmatter}
	
\section{Introduction}
\label{sec:introduction}

Constructing a low-cardinality discrete approximation of a positive measure
while preserving response-relevant information is a recurring task in
numerical analysis and scientific computing. Such constructions may be
interpreted as quadrature rules, in which a finely resolved distribution is
replaced by a small set of representative support points with associated
nonnegative weights. Probability-table construction in resonance
self-shielding provides a practically important instance of this general
problem.

In multigroup reactor calculations, resonance-energy cross sections vary
sharply within a single energy group. Strong absorption in resonance ranges
produces pronounced neutron flux depression, so that the effective group
cross sections depend sensitively on the local neutron spectrum and material
environment. One therefore seeks a reduced within-group representation that
remains accurate for the relevant responses while keeping the computational
cost under control. The subgroup method and related probability-table constructions
\cite{ChibaUnesaki2006,HebertCoste2002MomentTables,Hebert2005RibonExtended,Levitt1972ProbabilityTable},
together with subsequent methodological extensions
\cite{Hebert2009SubgroupProjection,LiZhangZhangZhao2020FineGroupSlowingDown,LiZhangLiuZhangHaoWangJiangLiu2023FineMeshSubgroup,RosierMaoZmijarevicLealSanchez2022PhysicalPT}
and application-oriented developments
\cite{StimpsonLiuCollinsClarno2017MPACT,LiuHeZuCaoWuZhang2018PRNSMGlobalLocal,LiuHeWenZuCaoWu2018PRNSMNECPX,YinZhangLiuHeZhang2025VITAS}
have been widely studied to address this need by replacing the
continuous within-group variation with a finite set of representative levels.
In an \(N\)-level probability table, each resonance group is represented by
subgroup levels, associated subgroup probabilities, and reaction-channel
levels. The effective cross section is then obtained as the
probability-weighted combination of the subgroup responses.

From a numerical perspective, this construction is naturally interpreted as a
quadrature problem in cross-section space. The probability table defines a
positive discrete measure supported on the subgroup levels, and the subgroup
transport calculation evaluates the target response at these support points.
The central task is therefore to construct a low-order discrete
representation that remains highly accurate over the target dilution range,
while preserving the nonnegative-real character of the transport-admissible
subgroup data and of the associated probabilities. Here accuracy concerns the quality of the effective-response approximation. Nonnegative-realness, by contrast, is an admissibility requirement throughout the construction: for the subgroup quantity entering the subgroup fixed-source transport calculation, for the subgroup probabilities as nonnegative discrete-measure weights, and for the effective cross sections supplied to the main multigroup transport calculation.

The classical moment-based construction of probability tables goes back to
Ribon \cite{RibonMaillard1986}, who matched a finite set of within-group
moments by an \(N\)-point discrete representation and recovered the subgroup
levels and probabilities through a Stieltjes--Pad\'e reconstruction. Chiba
\cite{ChibaUnesaki2006} retained this moment-based framework
while modifying the choice of preserved moments according to physical
considerations. These approaches are conceptually attractive because they
relate the discrete subgroup representation directly to the underlying
within-group distribution through a finite moment sequence. In floating-point
arithmetic, however, moment-based construction is often numerically fragile.
It typically requires explicit reconstruction from a truncated moment
sequence, the solution of ill-conditioned linear or Hankel-type systems, and
the nonlinear recovery of nodes and weights through root and residue
extraction. In difficult resonance groups, these stages can strongly amplify
roundoff perturbations, leading to loss of reality or nonnegativity in the
subgroup quantities or probabilities and, in turn, to nonphysical effective
cross sections. The main difficulty is therefore not only approximation
accuracy in exact arithmetic, but also the numerical stability of the
reconstruction procedure in finite precision.

Motivated by this difficulty, the present work reformulates Chiba's physically motivated affine moments \cite{ChibaUnesaki2006} as polynomial moments of a transformed positive measure. This yields an alternative construction route in which the subgroup rule is obtained by discrete-measure realization, symmetric Lanczos reduction, and Golub--Welsch extraction, rather than by explicit affine-moment inversion and Pad\'e-type root--residue recovery. The compression step preserves the real-nonnegative structure inherited from the underlying positive measure, and the reaction-channel levels are then reconstructed on the compressed nodes through coefficient matching in an orthogonal basis rather than through a Vandermonde solve.
In exact arithmetic, this route targets the same Gauss-type object as the classical Pad\'e construction for Stieltjes moment data \cite{Allen1974PadeGauss}; the distinction lies in the finite-precision procedure by which that target is constructed.

The contributions of this work are threefold. First, we reformulate Chiba's affine-order prescription as polynomial moments of a transformed positive measure, recasting probability-table construction as a structured positive-measure compression problem. Second, this reformulation yields a Lanczos--Golub--Welsch construction route that replaces the conventional moment--Pad\'e pipeline by discrete-measure realization, symmetric tridiagonal reduction, and orthogonal-basis reconstruction. Third, we present a response-level viewpoint that separates realization error from \(N\)-point compression error and clarifies the finite-precision behavior of the proposed construction.

Section~\ref{sec2} reviews the conventional moment--Pad\'e construction and its
finite-precision difficulties. Section~\ref{sec:method_lanczos} develops the transformed-measure
formulation and the corresponding Lanczos--Golub--Welsch construction.
Section~\ref{sec:results} examines the numerical behavior of the resulting method, and
Section~\ref{sec5} concludes the paper.

	\section{Problem setting and the conventional moment--Pad\'e construction}
  \label{sec2}

\subsection{Problem setting: induced state measure and subgroup observable}
\label{subsec:setting_subgroup}

Deterministic reactor-physics calculations typically employ a multigroup
discretization in neutron incident energy and require group-wise
effective microscopic cross sections as input to the transport solve.
For a reaction channel \(x\) in energy group
\(g:\;E\in [E_{g+1},E_g]\), the exact effective cross section is the
flux-weighted average
\begin{equation}
\sigma_{x,g}
\;\equiv\;
\frac{\displaystyle\int_{E_{g+1}}^{E_g} \sigma_x(E)\,\phi(E)\,dE}
{\displaystyle\int_{E_{g+1}}^{E_g} \phi(E)\,dE},
\label{eq:ribon_effxs_energy}
\end{equation}
where \(\phi(E)\) denotes the unknown within-group flux spectrum in the
target problem.

In resonance energy ranges, the within-group flux $\phi(E)$ exhibits
pronounced fine-energy structure due to resonance self-shielding.
Under the Bondarenko narrow-resonance approximation (NRA), the
scattering source and the background spectral shape vary slowly on the
scale of individual resonances, so that the energy dependence of the
flux is represented primarily through the self-shielding denominator.
Up to an energy-independent normalization,
\begin{equation}
\phi(E;\sigma_0)\propto \frac{1}{\sigma_t(E)+\sigma_0},
\qquad E\in [E_{g+1},E_g].
\label{eq:bondarenko_flux_form}
\end{equation}
where \(\sigma_{0}>0\) is a problem-dependent background cross section
parameter representing dilution and environment effects. Substituting
\eqref{eq:bondarenko_flux_form} into \eqref{eq:ribon_effxs_energy} yields the
NRA reference map
\begin{equation}
\sigma^{\mathrm{ref}}_{x,g}(\sigma_{0})
:=
\frac{\displaystyle\int_{E_{g+1}}^{E_g}\frac{\sigma_x(E)}{\sigma_t(E)+\sigma_{0}}\,dE}
{\displaystyle\int_{E_{g+1}}^{E_g}\frac{1}{\sigma_t(E)+\sigma_{0}}\,dE}.
\label{eq:nra_ref_eff_setting}
\end{equation}
Accordingly, the computational target in this work is the
dilution-response function
\(\sigma^{\mathrm{ref}}_{x,g}(\sigma_{0})\), evaluated on a prescribed
set of dilution points.

Equation \eqref{eq:bondarenko_flux_form} shows that, under the NRA, the energy
dependence enters through the resonant state \(\sigma_t(E)\). It is
therefore natural to reformulate the problem in cross-section space
rather than in energy space. Let \(\Delta E_g:=E_g-E_{g+1}\), and define
the uniform-in-energy probability measure
\[
\nu_g(dE):=\frac{1}{\Delta E_g}\,dE
\qquad\text{on }[E_{g+1},E_g].
\]
Assume that \(\sigma_t:[E_{g+1},E_g]\to (0,\infty)\) is Borel
measurable. The induced state measure on \((0,\infty)\) is the
pushforward
\[
\rho_g := (\sigma_t)_\# \nu_g,
\]
that is, for every bounded measurable test function \(\varphi\),
\begin{equation}
\int \varphi(\sigma)\,d\rho_g(\sigma)
=
\frac{1}{\Delta E_g}\int_{E_{g+1}}^{E_g}\varphi(\sigma_t(E))\,dE.
\label{eq:ribon_energy_to_sigma_compact}
\end{equation}
This formulation does not require the map \(E\mapsto \sigma_t(E)\) to be
monotone; any multi-branch resonance structure is absorbed into the
measure \(\rho_g\).

In general, the reaction-channel cross section \(\sigma_x(E)\) need not
be a single-valued function of \(\sigma_t(E)\). To avoid imposing such a
functional dependence, we introduce instead the finite
\(\sigma_x\)-weighted pushforward measure
\[
\rho^{(x)}_g := (\sigma_t)_\#(\sigma_x\,\nu_g),
\]
so that for every bounded measurable test function \(\psi\),
\begin{equation}
\int \psi(\sigma)\,d\rho^{(x)}_g(\sigma)
=
\frac{1}{\Delta E_g}\int_{E_{g+1}}^{E_g}\psi(\sigma_t(E))\,\sigma_x(E)\,dE.
\label{eq:ribon_weighted_pushforward}
\end{equation}

For \(\sigma_0>0\), define the denominator and numerator state-space
integrals
\begin{equation}
D_g(\sigma_0)
:=
\int \frac{1}{\sigma+\sigma_0}\,d\rho_g(\sigma),
\qquad
N_g^{(x)}(\sigma_0)
:=
\int \frac{1}{\sigma+\sigma_0}\,d\rho_g^{(x)}(\sigma).
\label{eq:setting_ratio_defs}
\end{equation}
Then the NRA reference effective cross section can be written as the
ratio
\begin{equation}
\sigma^{\mathrm{ref}}_{x,g}(\sigma_0)
=
\frac{N_g^{(x)}(\sigma_0)}{D_g(\sigma_0)}.
\label{eq:setting_ratio_ref}
\end{equation}
Thus, under the induced-measure formulation, the subgroup observable of
interest is naturally expressed not as a pointwise approximation problem
in energy, but as a ratio of two state-space integrals.

An \(N\)-level probability table with total levels
\(\{\sigma_{t,g,i}\}_{i=1}^N\), probabilities
\(\{p_{g,i}\}_{i=1}^N\), and reaction-channel levels
\(\{\sigma_{x,g,i}\}_{i=1}^N\) induces the atomic measures
\begin{equation}
\rho_{g,N}
=
\sum_{i=1}^N p_{g,i}\,\delta(\sigma-\sigma_{t,g,i}),
\qquad
\rho^{(x)}_{g,N}
=
\sum_{i=1}^N p_{g,i}\sigma_{x,g,i}\,\delta(\sigma-\sigma_{t,g,i}).
\label{eq:setting_atomic_measures}
\end{equation}
The corresponding subgroup prediction is therefore
\begin{equation}
\sigma^{\mathrm{PT}}_{x,g}(\sigma_0)
=
\frac{\displaystyle\int \frac{1}{\sigma+\sigma_0}\,d\rho^{(x)}_{g,N}(\sigma)}
{\displaystyle\int \frac{1}{\sigma+\sigma_0}\,d\rho_{g,N}(\sigma)}.
\label{eq:setting_ratio_pt}
\end{equation}
Evaluating the atomic integrals gives the familiar subgroup folding
formula
\begin{equation}
\sigma^{\mathrm{PT}}_{x,g}(\sigma_0)
=
\frac{\displaystyle\sum_{i=1}^{N} \frac{p_{g,i}\,\sigma_{x,g,i}}{\sigma_{t,g,i}+\sigma_{0}}}
{\displaystyle\sum_{i=1}^{N} \frac{p_{g,i}}{\sigma_{t,g,i}+\sigma_{0}}}.
\label{eq:nra_eff_setting}
\end{equation}

Equations \eqref{eq:nra_ref_eff_setting} and
\eqref{eq:nra_eff_setting} define the computational target adopted in
this work: construct subgroup parameters
\(\{(\sigma_{t,g,i},p_{g,i},\sigma_{x,g,i})\}_{i=1}^N\) such that
\(\sigma^{\mathrm{PT}}_{x,g}(\sigma_{0})\) reproduces the reference
response \(\sigma^{\mathrm{ref}}_{x,g}(\sigma_{0})\) over a prescribed
set of dilution points \(\sigma_{0}\in\mathcal{S}\). Throughout this
work,
\(\mathcal{S}=\{\sigma_{0,k}\}_{k=1}^{K}\)
denotes a logarithmically spaced grid spanning
\(\sigma_0\in[10^{-1},\,10^{6}]\)~barn.

\subsection{Ribon's moment--Pad\'e probability-table construction}
\label{subsec:ribon_pade}

The state-space formulation in Section~\ref{subsec:setting_subgroup}
identifies subgroup construction as the problem of replacing the induced
measure \(\rho_g\) by a low-cardinality atomic approximation. A classical
route, due to Ribon \cite{RibonMaillard1986}, constructs such an
approximation from a finite moment sequence through a Stieltjes--Pad\'e
reconstruction. Since this conventional route serves here only as the
baseline for later comparison and finite-precision analysis, we record
only the computational steps and notation needed below.

Within group \(g\), the total-cross-section state measure is approximated
by an \(N\)-point atomic measure
\begin{equation}
\rho_{g,N}(d\sigma)
:=
\sum_{i=1}^{N} p_{g,i}\,\delta(\sigma-\sigma_{t,g,i})\,d\sigma,
\qquad
p_{g,i}>0,
\quad
\sum_{i=1}^{N} p_{g,i}=1,
\label{eq:ribon_discrete_measure}
\end{equation}
so that, for suitable test functions \(f\),
\[
\int f(\sigma)\,d\rho_g(\sigma)
\approx
\sum_{i=1}^{N} p_{g,i}\,f(\sigma_{t,g,i}).
\]

Ribon's construction determines
\(\{(\sigma_{t,g,i},p_{g,i})\}_{i=1}^N\) from the moments of the induced
state measure. Define
\begin{equation}
m_k
:=
\int \sigma^k\,d\rho_g(\sigma)
=
\frac{1}{\Delta E_g}\int_{E_{g+1}}^{E_g}\sigma_t(E)^k\,dE,
\qquad
k=0,1,2,\dots,
\label{eq:ribon_moments}
\end{equation}
with \(m_0=1\). The atomic approximation \eqref{eq:ribon_discrete_measure}
is then required to match the first \(2N\) moments:
\begin{equation}
m_k
=
\sum_{i=1}^{N} p_{g,i}\,\sigma_{t,g,i}^{\,k},
\qquad
k=0,1,\dots,2N-1.
\label{eq:ribon_moment_matching}
\end{equation}
Equivalently, \(\rho_{g,N}\) is an \(N\)-point Gauss-type quadrature rule
for the state measure \(\rho_g\) \cite{Allen1974PadeGauss}.

Introduce the Stieltjes transform of \(\rho_g\),
\begin{equation}
G_g(z)
:=
\int \frac{1}{z-\sigma}\,d\rho_g(\sigma),
\qquad
z\in\mathbb{C}\setminus \operatorname{supp}(\rho_g),
\label{eq:ribon_stieltjes}
\end{equation}
whose Laurent expansion at \(z=\infty\) is
\begin{equation}
G_g(z)
=
\sum_{k=0}^{\infty}\frac{m_k}{z^{k+1}}.
\label{eq:ribon_stieltjes_series}
\end{equation}
For an \(N\)-point atomic measure,
\[
G_{g,N}(z)
=
\sum_{i=1}^{N}\frac{p_{g,i}}{z-\sigma_{t,g,i}},
\]
which is a rational function of type \([N\!-\!1/N]\). Ribon constructs
\(G_{g,N}\) as the \([N\!-\!1/N]\) Pad\'e approximant of the series
\eqref{eq:ribon_stieltjes_series}, that is,
\begin{equation}
G_g(z)-\frac{P_{N-1}(z)}{Q_N(z)}
=
\mathcal{O}\!\left(z^{-2N-1}\right),
\qquad
z\to\infty,
\label{eq:ribon_pade}
\end{equation}
where
\begin{equation}
Q_N(z)=z^N+q_1 z^{N-1}+\cdots+q_N,
\qquad
P_{N-1}(z)=a_0 z^{N-1}+a_1 z^{N-2}+\cdots+a_{N-1}.
\label{eq:ribon_QP_def}
\end{equation}

Equating coefficients in the Laurent expansion of \(Q_N(z)G_g(z)\)
yields the Hankel linear system for the denominator coefficients
\(q:=(q_1,\dots,q_N)^{\mathsf T}\):
\begin{equation}
\underbrace{\begin{bmatrix}
m_{1} & m_{2} & \cdots & m_{N} \\
m_{2} & m_{3} & \cdots & m_{N+1} \\
\vdots & \vdots & \ddots & \vdots \\
m_{N} & m_{N+1} & \cdots & m_{2N-1}
\end{bmatrix}}_{=:H_g}
\begin{bmatrix}
q_1 \\ q_2 \\ \vdots \\ q_N
\end{bmatrix}
=
-\begin{bmatrix}
m_{0} \\ m_{1} \\ \vdots \\ m_{N-1}
\end{bmatrix}.
\label{eq:ribon_hankel_matrix}
\end{equation}
Once \(Q_N\) is known, the numerator coefficients follow from the
polynomial part of \(Q_N(z)G_g(z)\), equivalently
\begin{equation}
a_k=\sum_{j=0}^{k} q_j\,m_{k-j},
\qquad
k=0,1,\dots,N-1,
\label{eq:ribon_P_from_m}
\end{equation}
with \(q_0=1\).

The subgroup total levels are then obtained as the roots of \(Q_N\):
\begin{equation}
Q_N(\sigma_{t,g,i})=0,
\qquad
i=1,\dots,N.
\label{eq:ribon_nodes}
\end{equation}
Assuming that these roots are simple, the subgroup probabilities are
recovered from the corresponding residues,
\begin{equation}
p_{g,i}
=
\operatorname*{Res}_{z=\sigma_{t,g,i}}
\frac{P_{N-1}(z)}{Q_N(z)}
=
\frac{P_{N-1}(\sigma_{t,g,i})}{Q_N'(\sigma_{t,g,i})},
\qquad
i=1,\dots,N.
\label{eq:ribon_weights}
\end{equation}
In exact arithmetic, when \(\{m_k\}\) is an exact Stieltjes moment
sequence of a positive measure, this reconstruction is consistent with a
positive \(N\)-point atomic rule \cite{Allen1974PadeGauss}.

Once \(\{(\sigma_{t,g,i},p_{g,i})\}_{i=1}^N\) have been recovered, the
subgroup reaction-channel levels are determined from mixed moments.
Define
\begin{equation}
d_k
:=
\frac{1}{\Delta E_g}\int_{E_{g+1}}^{E_g}
\sigma_x(E)\,\sigma_t(E)^k\,dE,
\qquad
k=0,1,\dots,N-1.
\label{eq:ribon_mixed_moments}
\end{equation}
The unknown partial levels \(\{\sigma_{x,g,i}\}_{i=1}^N\) are then
obtained by enforcing
\begin{equation}
d_k
=
\sum_{i=1}^{N}
p_{g,i}\,\sigma_{x,g,i}\,\sigma_{t,g,i}^{\,k},
\qquad
k=0,1,\dots,N-1.
\label{eq:ribon_partial_match}
\end{equation}
This is a Vandermonde-type linear system for the subgroup
reaction-channel levels. With all subgroup parameters determined, the
effective cross section is finally computed from the folding formula
\eqref{eq:nra_eff_setting}.

The construction in this subsection is classical and conceptually natural: it recovers an \(N\)-point atomic rule from moment information through Stieltjes--Pad\'e reconstruction \cite{AptekarevBuslaevMartinezFinkelshteinSuetin2011Pade}. The issue addressed next is therefore not its exact-arithmetic consistency, but its numerical stability in finite precision.

\subsection{Finite-precision instability and score-based diagnosis of the conventional moment--Pad\'e construction}
\label{sec:illposedness_pade}

In exact arithmetic, Pad\'e--Stieltjes reconstruction from exact Stieltjes
moment data recovers an \(N\)-point Gauss-type quadrature for a positive
measure and hence yields real nodes and positive weights. In floating-point
arithmetic, however, the conventional moment--Pad\'e construction is a
multi-stage sensitivity problem
\cite{BakerGravesMorris1996,Hopkins1982,Bai2002}.
The present subgroup construction inherits the same sensitivity mechanisms and
further appends a mixed-moment reconstruction of the reaction-channel levels on
the recovered nodes. Unless otherwise stated, all main-text numerical and
diagnostic results in this subsection are for the
\(^{238}\mathrm{U}\) capture case under a fine-group structure.

For each resonance group, the conventional construction realizes the map
\[
\{m_k\}_{k=0}^{2N-1}
\longmapsto
\{(\sigma_{t,g,i},\,p_{g,i},\,\sigma_{x,g,i})\}_{i=1}^N .
\]
In finite precision, the corresponding perturbation chain is
\[
\delta m \longrightarrow \delta q \longrightarrow (\delta \sigma_t,\delta p)
\longrightarrow \delta \sigma_x ,
\]
where \(\delta m\) denotes perturbations in the preserved moments,
\(\delta q\) the induced perturbation in the Pad\'e denominator coefficients,
\((\delta \sigma_t,\delta p)\) the perturbations in the recovered subgroup
levels and probabilities, and \(\delta \sigma_x\) the perturbation in the
reconstructed partial reaction levels.

At the first stage, the \([N\!-\!1/N]\) Pad\'e condition leads to the Hankel
system \eqref{eq:ribon_hankel_matrix} for the denominator coefficients
\(\bm q=(q_1,\ldots,q_N)^{\mathsf T}\). A standard first-order estimate gives
\begin{equation}
\frac{\|\delta\bm q\|}{\|\bm q\|}
\;\lesssim\;
\kappa_2(H_g)
\left(
\frac{\|\delta H_g\|}{\|H_g\|}
+
\frac{\|\delta \bm h_g\|}{\|\bm h_g\|}
\right),
\label{eq:pade_q_sensitivity_revised}
\end{equation}
where \(H_g\) is the groupwise Hankel moment matrix and
\(\bm h_g\) is the corresponding shifted moment vector.

At the second stage, the subgroup total levels are recovered as the roots of
the Pad\'e denominator polynomial \(Q\), and the subgroup probabilities from
the corresponding residues. If \(\sigma_i\) is a simple root of \(Q\) and the
polynomial is perturbed by \(\delta Q\), then
\begin{equation}
\delta \sigma_i
\;\approx\;
-\frac{\delta Q(\sigma_i)}{Q'(\sigma_i)}.
\label{eq:root_first_order_revised}
\end{equation}
Hence clustered or near-multiple roots, for which \(|Q'(\sigma_i)|\) is small,
can induce large node perturbations under small coefficient perturbations.
Residue recovery inherits the same sensitivity. This is consistent with the
pole sensitivity emphasized in the modern Pad\'e literature, where finite-precision
difficulties often appear through near-singularity and spurious pole--zero
pairs \cite{GonnetGuettelTrefethen2013}. In the present setting, the practical manifestation is loss of real,
nonnegative subgroup data.

At the third stage, the partial reaction levels are obtained from a
Vandermonde-type mixed-moment system. This introduces a further amplification
mechanism, since irregular or poorly separated nodes typically lead to
ill-conditioned Vandermonde matrices. Negative or complex subgroup quantities
are therefore the cumulative effect of this three-stage perturbation cascade.

These observations motivate the diagnostic scores
\begin{align}
s_g^{(\mathrm{sep})}
&=
-\log_{10}\!\Big(\min_{i\ne j}|\sigma_i-\sigma_j|\Big),
\label{eq:score_sep}
\\
s_g^{(\sigma)}
&=
-\log_{10}\!\Big(\min_i \big|(\Pi_g^{(\sigma)})'(\sigma_i)\big|\Big),
\qquad
\Pi_g^{(\sigma)}(x):=\prod_{i=1}^N (x-\sigma_i),
\label{eq:score_sigma}
\\
s_g^{(\rho)}
&=
-\log_{10}\!\Big(\min_i \big|(\Pi_g^{(\rho)})'(\rho_i)\big|\Big),
\qquad
\Pi_g^{(\rho)}(y):=\prod_{i=1}^N (y-\rho_i),
\label{eq:score_rho}
\\
s_g^{(H)}
&=
\log_{10}\kappa_2(H_g),
\label{eq:score_H}
\end{align}
together with the comparison score
\begin{equation}
s_g^{(V)}=\log_{10}\kappa_2(V_g).
\label{eq:score_V}
\end{equation}
Here \(s_g^{(\mathrm{sep})}\) measures node clustering,
\(s_g^{(\sigma)}\) and \(s_g^{(\rho)}\) measure first-order root sensitivity
in the \(\sigma\)- and transformed-space representations, respectively, and
\(s_g^{(H)}\) measures first-stage amplification through the Hankel condition
number.

At fixed order \(N\), groups are partitioned into
nonnegative-real-preserving and nonnegative-real-violating classes. For any
diagnostic score \(s_g\), discriminatory power is quantified by the receiver
operating characteristic (ROC) curve and, in particular, by the area under this
curve (AUC),
\begin{equation}
\mathrm{AUC}
=
\frac{U}{n_{\mathrm{viol}}\,n_{\mathrm{pres}}},
\label{eq:auc_def}
\end{equation}
where \(U\) is the Mann--Whitney statistic and
\(n_{\mathrm{viol}}, n_{\mathrm{pres}}\) are the class counts. We also record
bootstrap 95\% confidence intervals for the AUC values. Only nondegenerate
pre-collapse orders, at which both classes are present, are included in the AUC
interpretation.

\begin{figure}[htbp]
  \centering
  \includegraphics[width=0.72\linewidth]{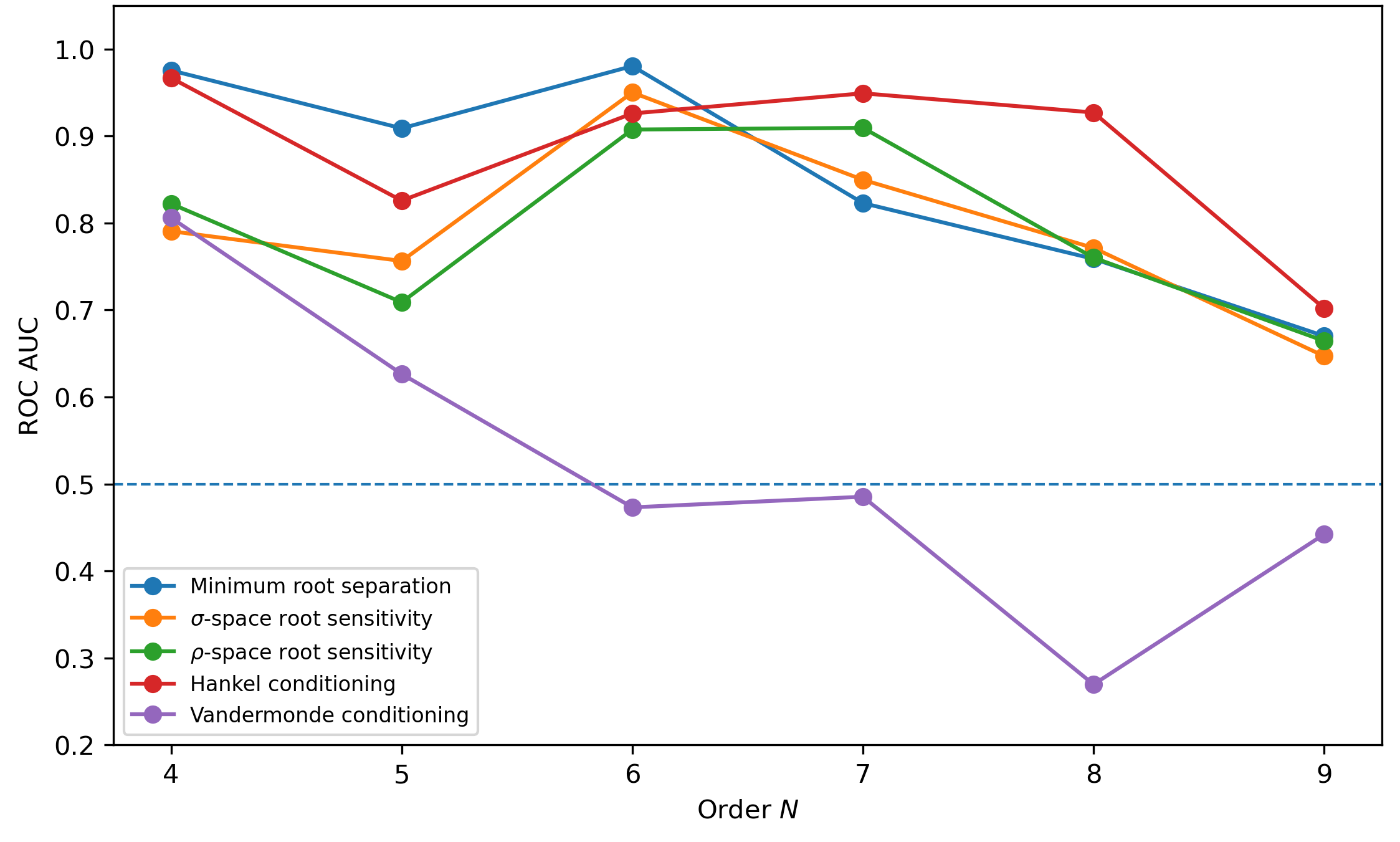}
  \caption{
    ROC AUC summary for the diagnostic scores $s_g^{(\mathrm{sep})}$,
    $s_g^{(\sigma)}$, $s_g^{(\rho)}$, $s_g^{(H)}$, and $s_g^{(V)}$ defined in
    \eqref{eq:score_sep}--\eqref{eq:score_V}, shown as functions of reconstruction order $N$. The dashed
    horizontal line at $\mathrm{AUC}=0.5$ marks chance-level separation. AUC values are shown only for nondegenerate pre-collapse orders at which both
    preserving and violating groups are present; the corresponding class counts are reported separately in the
    Supplementary Material. Bootstrap 95\% confidence intervals for the four
    principal scores are also reported there. The minimum root-separation score, the $\sigma$-space and
    $\rho$-space root-sensitivity scores, and the Hankel-conditioning score retain
    substantial discriminatory power over the tested orders, whereas the
    Vandermonde-conditioning score deteriorates rapidly and approaches chance level,
    or falls below it, at moderate and high orders.
  }
  \label{fig:diagnostic-auc-summary}
\end{figure}

Figure~\ref{fig:diagnostic-auc-summary} shows that the four principal scores
\eqref{eq:score_sep}--\eqref{eq:score_H} retain nontrivial discriminatory power
over the nondegenerate pre-collapse orders, whereas the Vandermonde score
\eqref{eq:score_V} is substantially weaker. The most informative signatures are
concentrated in the Hankel solve and the subsequent root-recovery stage.

A representative fixed-order view is given in
Figure~\ref{fig:diagnostic-boxplot-N7}, which shows the groupwise distributions
of the four principal scores at \(N=7\).
Figure~\ref{fig:diagnostic-boxplot-N7} shows a systematic upward shift of the
violating groups relative to the preserving groups in all four panels. The
clearest separation is observed for the \(\rho\)-space root-sensitivity and
Hankel-conditioning scores.

\begin{figure}[htbp]
  \centering
  \includegraphics[width=\linewidth]{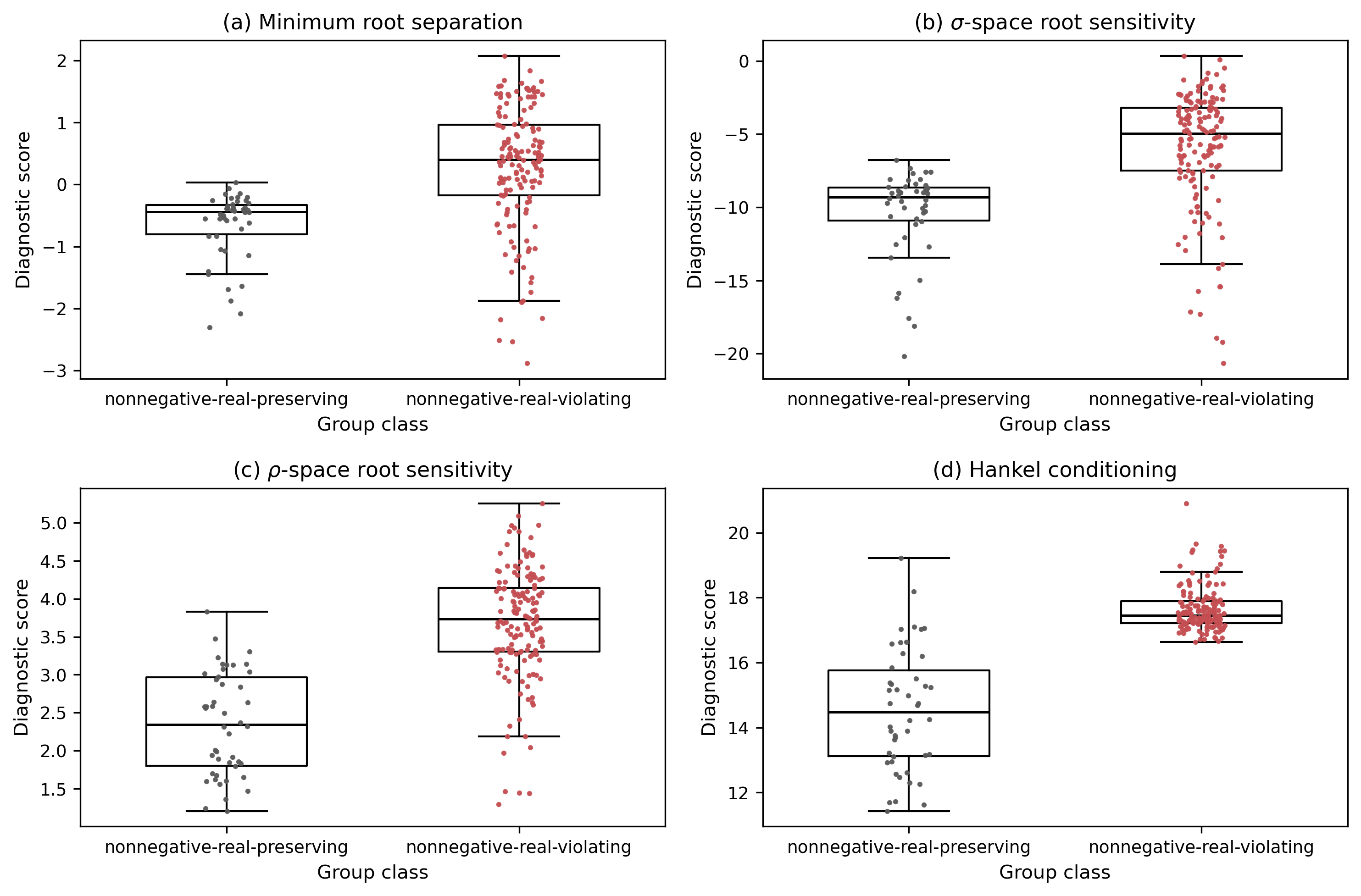}
  \caption{
    Distributions of the four principal diagnostic scores at $N=7$ for
    nonnegative-real-preserving and nonnegative-real-violating groups:
    (a) minimum root separation, (b) $\sigma$-space root sensitivity,
    (c) $\rho$-space root sensitivity, and (d) Hankel conditioning. Each point
    represents one energy group. This order lies in the mixed regime, so that both classes are represented by substantial samples. Larger scores indicate
    greater numerical danger.
  }
  \label{fig:diagnostic-boxplot-N7}
\end{figure}

In summary, the numerical difficulty does not lie in the exact-arithmetic correspondence between Pad\'e--Stieltjes reconstruction and Gauss quadrature itself, but in the conditioning of the moment-based representation and the sensitivity of the floating-point reconstruction maps. This diagnosis motivates the alternative route developed
next, which departs from the conventional explicit-moment Pad\'e pipeline and
instead proceeds through positive-measure realization and orthogonal reduction
\cite{Bai2002,GanderKarp2001}.

\section{Proposed method: Lanczos--Golub--Welsch probability table construction}
\label{sec:method_lanczos}

The key step of the present construction is to make explicit that the Chiba Case-C affine-order moments can be equivalently viewed as polynomial moments of a transformed positive measure. This reformulation yields an alternative constructive route for probability-table generation. The proposed construction does not proceed through direct inversion of the affine moment sequence. Instead, it first realizes the transformed measure on a computable discrete support and then applies Gauss compression to the realized positive measure. This realized positive-measure compression scheme provides the basis for the Lanczos--Golub--Welsch construction presented below.
Section~\ref{sec:method_lanczos} first introduces the transformed measure and its discrete realization, then studies realization and compression errors at the measure level, next establishes the required nonnegative-real properties and stable channel reconstruction, and finally connects these results to subgroup-observable stability via the ratio bound \eqref{eq:ratio_stability_bound}.

\subsection{Measure formulation and discrete realization}
\label{subsec:measure_uniongrid_casec}

Section~\ref{subsec:setting_subgroup} showed that the subgroup observable is a
ratio of state-space integrals. The proposed construction is based on a
transformed positive-measure formulation of Chiba's affine-order moments. The
present subsection introduces that measure, defines its computable discrete
realization on the adopted pointwise energy representation, and then separates
realization and compression errors at both the measure and observable levels.

Fix an energy group \(g:\,E\in [E_L,E_R]\), and endow it with the
uniform-in-energy probability measure
\[
\nu_g(dE):=\frac{1}{\Delta E_g}\,dE,
\qquad
\Delta E_g:=E_R-E_L.
\]
Let \(\sigma_t:[E_L,E_R]\to (0,\infty)\) denote the reconstructed pointwise
total microscopic cross section. As in
Section~\ref{subsec:setting_subgroup}, the mapping
\(E\mapsto \sigma_t(E)\) induces the state distribution on \((0,\infty)\) as
the pushforward measure
\[
\rho_g := (\sigma_t)_\# \nu_g,
\]
that is,
\begin{equation}
\int \varphi(\sigma)\,d\rho_g(\sigma)
=
\frac{1}{\Delta E_g}\int_{E_L}^{E_R}
\varphi\!\bigl(\sigma_t(E)\bigr)\,dE,
\label{eq:rho_pushforward_identity_clean_new}
\end{equation}
for every bounded Borel-measurable test function \(\varphi\).

We now consider the affine-order sequence of Chiba Case C
\cite{ChibaUnesaki2006},
\[
n_k = a + bk,
\qquad
k=0,1,\ldots,2N-1,
\]
with
\[
a=-1,
\qquad
b=\frac{9}{4(2N-1)}.
\]
For each prescribed subgroup order $N$, the associated Chiba Case-C affine-order sequence determines a corresponding transformed variable $z=\sigma_t^{\,b_N}$ and hence an $N$-dependent transformed measure, which we continue to denote by $\mu_g$ when no confusion can arise.

To encode these moments in a form compatible with positive-measure compression,
introduce the transformed variable
\[
z=\sigma^b,
\]
and define the unnormalized positive Borel measure \(\mu_g\) on \((0,\infty)\)
by
\begin{equation}
\mu_g := (z)_\#\bigl(\sigma^a \rho_g\bigr),
\qquad\text{equivalently}\qquad
\int f(z)\,d\mu_g(z)
=
\frac{1}{\Delta E_g}\int_{E_L}^{E_R}
f\!\bigl(\sigma_t(E)^b\bigr)\,\sigma_t(E)^a\,dE,
\label{eq:mu_energy_form_clean_new}
\end{equation}
for every bounded Borel-measurable test function \(f\). In particular, with
\(f(z)=z^k\),
\[
\int z^k\,d\mu_g(z)
=
\frac{1}{\Delta E_g}\int_{E_L}^{E_R}
\sigma_t(E)^{a+bk}\,dE,
\qquad
k=0,1,\ldots,2N-1.
\]
Thus polynomial moments of \(\mu_g\) reproduce exactly the affine-order moments
prescribed by Chiba Case C.

Let
\[
c_0:=\int 1\,d\mu_g>0,
\]
and define the normalized probability measure
\begin{equation}
\mu_{g,\mathrm{norm}}
:=
\frac{1}{c_0}\,\mu_g.
\label{eq:mu_norm_def_clean_new}
\end{equation}
The transformed-measure construction requires only mild regularity: it is
sufficient that \(\sigma_t\) be Borel measurable, positive almost everywhere on
\([E_L,E_R]\), and that
\[
\int_{E_L}^{E_R}\sigma_t(E)^a\,dE < \infty .
\]
Under these assumptions, \(\mu_g\) is a finite positive Borel measure on
\((0,\infty)\), and its polynomial moments are well defined whenever
\[
\int_{E_L}^{E_R}\sigma_t(E)^{a+bk}\,dE < \infty .
\]
For Chiba Case C, \(a=-1\), so the additional requirement is the integrability
of \(\sigma_t^{-1}\) on the group interval. In the discrete union-grid
realization used below, we impose the stronger computational assumption that
the retained union-grid segments have positive and finite endpoint values of
\(\sigma_t\), so that the nodal quantities introduced below are well defined and
finite.

In computation, \(\mu_{g,\mathrm{norm}}\) is first realized on the adopted
pointwise energy representation. Let \(\{E_j\}_{j=0}^{J}\) be the sorted union
of all tabulation points in \([E_L,E_R]\), including both group boundaries. A
normalized within-group energy integral is then approximated by the trapezoidal
average
\[
\frac{1}{\Delta E_g}\int_{E_L}^{E_R} h(E)\,dE
\;\approx\;
\sum_{j=0}^{J}\omega_j\,h(E_j),
\qquad
\omega_j\ge 0,
\qquad
\sum_{j=0}^{J}\omega_j=1.
\]
Define
\[
\pi_j:=\omega_j\,\sigma_t(E_j)^a,
\qquad
z_j:=\sigma_t(E_j)^b,
\qquad
w_j:=\pi_j\Big/\sum_{\ell=0}^{J}\pi_\ell.
\]
This yields the computable discrete probability measure
\begin{equation}
\mu_M
:=
\sum_{j=0}^{J} w_j\,\delta_{z_j},
\qquad
w_j>0,
\qquad
\sum_{j=0}^{J} w_j=1.
\label{eq:mu_M_def_clean_new}
\end{equation}
By construction, \(\mu_M\) is a positivity-preserving discrete realization of
\(\mu_{g,\mathrm{norm}}\). Its role is to provide a computable positive input
measure from which the subsequent \(N\)-point compression is constructed,
without explicit inversion of the affine moment sequence.

For later use, let
\[
\mu_N
:=
\sum_{i=1}^{N} p_i\,\delta_{\widehat z_i},
\qquad
p_i>0,
\qquad
\sum_{i=1}^{N}p_i=1,
\]
denote an \(N\)-point compression of \(\mu_M\). For any bounded
Borel-measurable test function \(f\), define
\[
\begin{aligned}
I(f)   &:= \int f(z)\,d\mu_{g,\mathrm{norm}}(z), \qquad
I_M(f) := \int f(z)\,d\mu_M(z)=\sum_{j=0}^{J} w_j f(z_j), \\
I_N(f) &:= \int f(z)\,d\mu_N(z)=\sum_{i=1}^{N} p_i f(\widehat z_i).
\end{aligned}
\]
Then
\begin{equation}
I(f)-I_N(f)
=
\bigl(I(f)-I_M(f)\bigr)
+
\bigl(I_M(f)-I_N(f)\bigr).
\label{eq:error_split}
\end{equation}
Equation~\eqref{eq:error_split} separates the total error into a realization
term, arising from replacement of \(\mu_{g,\mathrm{norm}}\) by the computable
discrete measure \(\mu_M\), and a compression term, arising from reduction of
\(\mu_M\) to the \(N\)-point rule \(\mu_N\).

The subgroup observable of interest, however, is a ratio of two state-space
integrals rather than a single linear functional; see
\eqref{eq:setting_ratio_ref}--\eqref{eq:setting_ratio_pt}. The following elementary ratio estimate transfers numerator and denominator errors to the final subgroup response.

Let
\[
R_g^{(x)}(\sigma_0):=\sigma^{\mathrm{ref}}_{x,g}(\sigma_0)
=\frac{N_g^{(x)}(\sigma_0)}{D_g(\sigma_0)},
\qquad
R_{g,N}^{(x)}(\sigma_0):=\sigma^{\mathrm{PT}}_{x,g}(\sigma_0)
=\frac{N_{g,N}^{(x)}(\sigma_0)}{D_{g,N}(\sigma_0)}.
\]
Likewise, let \(D_{g,M}\) and \(N_{g,M}^{(x)}\) denote the corresponding
denominator and numerator associated with the discrete realization \(\mu_M\).

To pass from the measure-level decomposition to the observable level, we use
the following elementary ratio estimate. Fix \(g\), \(x\), and \(\sigma_0>0\),
and write
\[
R_g^{(x)}:=\frac{N_g^{(x)}}{D_g},
\qquad
R_{g,N}^{(x)}:=\frac{N_{g,N}^{(x)}}{D_{g,N}}.
\]
Assume that \(D_g>0\), and define
\[
\varepsilon_D:=|D_{g,N}-D_g|,
\qquad
\varepsilon_N:=|N_{g,N}^{(x)}-N_g^{(x)}|.
\]
If \(\varepsilon_D<D_g\), then
\begin{equation}
\bigl|R_{g,N}^{(x)}-R_g^{(x)}\bigr|
\le
\frac{\varepsilon_N+|R_g^{(x)}|\,\varepsilon_D}{D_g-\varepsilon_D}.
\label{eq:ratio_stability_bound}
\end{equation}
Indeed,
\[
R_{g,N}^{(x)}-R_g^{(x)}
=
\frac{N_{g,N}^{(x)}}{D_{g,N}}-\frac{N_g^{(x)}}{D_g}
=
\frac{
D_g\bigl(N_{g,N}^{(x)}-N_g^{(x)}\bigr)
-
N_g^{(x)}\bigl(D_{g,N}-D_g\bigr)
}{
D_g\,D_{g,N}
},
\]
so, using \(N_g^{(x)}=R_g^{(x)}D_g\),
\[
\bigl|R_{g,N}^{(x)}-R_g^{(x)}\bigr|
\le
\frac{\varepsilon_N+|R_g^{(x)}|\,\varepsilon_D}{|D_{g,N}|}.
\]
Moreover,
\[
|D_{g,N}|
\ge
D_g-|D_{g,N}-D_g|
=
D_g-\varepsilon_D
>
0,
\]
which yields \eqref{eq:ratio_stability_bound}.

Equation~\eqref{eq:ratio_stability_bound} reduces observable-level accuracy to control of
the numerator and denominator errors, provided the denominator remains bounded
away from zero. Since the dilution set \(\mathcal S\) used in this work is
finite,
\[
d_{*,g}:=\min_{\sigma_0\in\mathcal S} D_g(\sigma_0)>0,
\]
and the same estimate holds uniformly over \(\sigma_0\in\mathcal S\). Moreover,
by the triangle inequality,
\begin{align}
\varepsilon_D(\sigma_0)
&\le
\bigl|D_g(\sigma_0)-D_{g,M}(\sigma_0)\bigr|
+
\bigl|D_{g,M}(\sigma_0)-D_{g,N}(\sigma_0)\bigr|,
\label{eq:denominator_two_stage_split}
\\
\varepsilon_N(\sigma_0)
&\le
\bigl|N_g^{(x)}(\sigma_0)-N_{g,M}^{(x)}(\sigma_0)\bigr|
+
\bigl|N_{g,M}^{(x)}(\sigma_0)-N_{g,N}^{(x)}(\sigma_0)\bigr|.
\label{eq:numerator_two_stage_split}
\end{align}
The first term in each decomposition is the realization error associated with
discretization of the induced measure, while the second is the additional
compression error introduced by the \(N\)-point reduction. Thus the
effective-cross-section error inherits the same two-stage structure as
\eqref{eq:error_split}. Once the realization term dominates, refining the
\(N\)-point compression alone cannot materially improve the observable.

In the present work, two realization strategies are used for reference
evaluation and for constructing the discrete input measure. The first is a
trapezoidal union-grid realization, in which \(\mu_M\) is constructed directly
on the union grid with nonnegative normalized trapezoidal weights. Its main
role is to provide a structurally matched baseline: the same nodal
representation is used in the reference evaluation and in the input discrete
measure, so that the compression term in \eqref{eq:error_split} can be assessed
with minimal additional mismatch from the realization itself.

The second realization follows the piecewise-linear interpolation of the
adopted ENDF/B-VIII.1 \(^{238}\mathrm{U}\) total and capture cross-section
data. Here the reference response is evaluated by segmentwise analytic
integration, and the input discrete measure is constructed by Gauss--Legendre
quadrature on each union-grid segment. This variant is used as a realization
benchmark tied directly to the adopted ENDF piecewise-linear interpolation, but
not necessarily as a closer approximation to the underlying physics beyond that
model.

\subsection{Exact-arithmetic Gauss compression and subgroup reconstruction}
\label{subsec:exact_gauss_compression}

We now address the compression step in the decomposition
\eqref{eq:error_split}. Starting from the realized positive discrete measure
\[
\mu_M(dz)=\sum_{j=0}^{J} w_j\,\delta(z-z_j)\,dz,
\qquad
w_j>0,
\qquad
\sum_{j=0}^{J} w_j=1,
\]
the objective is to construct an \(N\)-point compressed rule that preserves the
essential moment information of \(\mu_M\) while retaining the positivity
structure of the scalar quadrature variables. This subsection concerns that
construction in exact arithmetic. The subsequent reconstruction of
reaction-channel levels is then carried out on the fixed compressed nodes and
weights, but its role is to preserve prescribed mixed-moment information rather
than to provide a componentwise positivity theorem.

Introduce the diagonal matrix and normalized starting vector
\[
A:=\mathrm{diag}(z_0,\ldots,z_J),
\qquad
v_0:=(\sqrt{w_0},\ldots,\sqrt{w_J})^{\mathsf T}.
\]
Then, for every polynomial \(f\),
\begin{equation}
v_0^{\mathsf T} f(A) v_0
=
\sum_{j=0}^{J} w_j f(z_j)
=
\int f(z)\,d\mu_M(z),
\label{eq:quad_form_identity}
\end{equation}
so that the pair \((A,v_0)\) realizes the moment functional of \(\mu_M\) on the
polynomial space.

Applying \(N\) steps of the symmetric Lanczos process to \((A,v_0)\) generates
an orthonormal basis
\[
V_N=[v_0,v_1,\ldots,v_{N-1}]
\]
of the Krylov subspace
\[
\mathcal K_N(A,v_0)=\mathrm{span}\{v_0,Av_0,\ldots,A^{N-1}v_0\},
\]
together with the three-term recurrence
\[
Av_k=\beta_k v_{k-1}+\alpha_k v_k+\beta_{k+1}v_{k+1},
\qquad
k=0,\ldots,N-1,
\]
with \(\beta_0=0\). Equivalently, the projected matrix
\[
J_N:=V_N^{\mathsf T}AV_N
\]
is the real symmetric tridiagonal Jacobi matrix associated with the realized measure \cite{Gautschi1982GeneratingOrthogonalPolynomials},
\begin{equation}
J_N=
\begin{pmatrix}
\alpha_0 & \beta_1 \\
\beta_1  & \alpha_1 & \ddots \\
         & \ddots   & \ddots & \beta_{N-1}\\
         &          & \beta_{N-1} & \alpha_{N-1}
\end{pmatrix}.
\label{eq:Jn_def}
\end{equation}

\begin{proposition}
\label{prop:gauss_compression}
Assume that the \(N\)-step Lanczos process applied to \((A,v_0)\) proceeds
without early breakdown, and let
\[
J_N = Q\Lambda Q^{\mathsf T},
\qquad
\Lambda=\mathrm{diag}(\lambda_1,\ldots,\lambda_N),
\]
be its eigendecomposition, where the columns \(q_i\) of \(Q\) are orthonormal
eigenvectors. Define
\[
\widehat z_i:=\lambda_i,
\qquad
p_i:=(e_1^{\mathsf T}q_i)^2,
\qquad
i=1,\ldots,N,
\]
with \(e_1=(1,0,\ldots,0)^{\mathsf T}\in\mathbb R^N\), and set
\[
\mu_N(dz)=\sum_{i=1}^{N} p_i\,\delta(z-\widehat z_i)\,dz.
\]
Then \(\mu_N\) is the \(N\)-point Gauss quadrature rule associated with
\(\mu_M\), namely
\begin{equation}
\int \pi(z)\,d\mu_N(z)
=
\int \pi(z)\,d\mu_M(z),
\qquad
\forall\,\pi\in\mathbb P_{2N-1}.
\label{eq:gauss_exactness}
\end{equation}
Moreover, the nodes \(\widehat z_i\) are real and lie in the convex hull of
\(\operatorname{supp}(\mu_M)\subset(0,\infty)\), while the weights satisfy
\[
p_i>0,
\qquad
\sum_{i=1}^{N}p_i=1.
\]
If \(z=\sigma_t^b\) with \(b>0\), then the mapped total subgroup levels
\[
\sigma_{t,i}=\widehat z_i^{1/b}
\]
are real and nonnegative.
\end{proposition}

\begin{proof}
Equation \eqref{eq:gauss_exactness} is the classical Lanczos--Golub--Welsch
Gauss quadrature result
\cite{GolubWelsch1969,GolubMeurant2010,Meurant2006LanczosCG,GolubVanLoan2013}.
Since \(J_N\) is real symmetric tridiagonal, its eigenvalues are real. For a
positive measure with support contained in \((0,\infty)\), the Gauss nodes lie
in the convex hull of the support, hence \(\widehat z_i>0\). The quadrature
weights are
\[
p_i=(e_1^{\mathsf T}q_i)^2>0,
\qquad
\sum_{i=1}^{N}p_i
=
e_1^{\mathsf T}QQ^{\mathsf T}e_1
=
1.
\]
Finally, because \(b>0\) and \(\widehat z_i>0\), the mapped levels
\(\sigma_{t,i}=\widehat z_i^{1/b}\) are real and nonnegative.
\end{proof}

\begin{remark}
\label{rem:scope_structure_preserving}
The structure-preserving statement in
Proposition~\ref{prop:gauss_compression} pertains specifically to the
positive-measure compression step, that is, to the reality and nonnegativity of
the compressed quadrature nodes and weights
\(\{(\widehat z_i,p_i)\}_{i=1}^N\), or equivalently of the mapped pairs
\(\{(\sigma_{t,i},p_i)\}_{i=1}^N\). It does not, in general, extend to the
subsequent reconstruction of reaction-channel levels or to componentwise
nonnegativity of the final effective responses.
\end{remark}

If early Lanczos breakdown occurs before step \(N\), then the realized measure
has reached its effective support dimension and an exact lower-order Gauss rule
has already been obtained. Thus early breakdown indicates an exact lower-rank
compression rather than a failure of the measure formulation.

To reconstruct the reaction-channel levels while keeping the same compressed
nodes and weights, let \(\{P_k\}_{k=0}^{N-1}\) denote the orthonormal
polynomials associated with \(\mu_M\), and define the orthogonal-basis
coefficients
\begin{equation}
b_k
=
\int \sigma_x(z)\,P_k(z)\,d\mu_M(z),
\qquad
k=0,\ldots,N-1.
\label{eq:bk_def}
\end{equation}
In the discrete union-grid realization, the Lanczos vectors satisfy
\[
v_k(j)=P_k(z_j)\sqrt{w_j},
\qquad
j=0,\ldots,J,
\]
so \eqref{eq:bk_def} can be evaluated through the stable inner products
\begin{equation}
b_k
=
\sum_{j=0}^{J}\sigma_x(z_j)\,P_k(z_j)\,w_j
=
\sum_{j=0}^{J}\sigma_x(z_j)\,v_k(j)\sqrt{w_j}.
\label{eq:bk_inner_product}
\end{equation}

Once the compressed Gauss nodes and weights
\(\{(\widehat z_i,p_i)\}_{i=1}^N\) have been fixed by
Proposition~\ref{prop:gauss_compression}, the remaining task is to reconstruct
the reduced reaction-channel levels on this fixed compressed rule. We define
\(\{\sigma_{x,i}\}_{i=1}^{N}\) by requiring that the first \(N\)
orthogonal-basis coefficients of the reaction-channel data be matched under the
compressed rule:
\begin{equation}
b_k
=
\sum_{i=1}^{N} p_i\,\sigma_{x,i}\,P_k(\widehat z_i),
\qquad
k=0,\ldots,N-1.
\label{eq:partial_matching}
\end{equation}

\begin{lemma}
\label{lemma:orthogonal_reconstruction}
Let \(J_N=Q\Lambda Q^{\mathsf T}\) be the eigendecomposition in
Proposition~\ref{prop:gauss_compression}, and let the reduced
reaction-channel levels be defined by \eqref{eq:partial_matching}. Then
\begin{equation}
b_k
=
\sum_{i=1}^{N}(Q_{ki}Q_{0i})\,\sigma_{x,i},
\qquad
k=0,\ldots,N-1,
\label{eq:Q_reconstruction}
\end{equation}
or equivalently
\[
Ms=b,
\qquad
M_{ki}=Q_{ki}Q_{0i},
\]
where
\[
s=(\sigma_{x,1},\ldots,\sigma_{x,N})^{\mathsf T},
\qquad
b=(b_0,\ldots,b_{N-1})^{\mathsf T}.
\]
Moreover,
\[
M = Q\,\mathrm{diag}(Q_{01},\ldots,Q_{0N}),
\]
so \(M\) is nonsingular whenever \(p_i=Q_{0i}^2>0\). Hence the reduced
reaction-channel levels are uniquely determined, and the reconstruction can be
carried out entirely in the orthogonal Lanczos basis, without solving a
Vandermonde system on the compressed nodes \(\{\widehat z_i\}\).
\end{lemma}

\begin{proof}
For the Gauss rule associated with \(J_N\), one has
\[
p_i=(e_1^{\mathsf T}q_i)^2=Q_{0i}^2,
\qquad
P_k(\widehat z_i)=\frac{Q_{ki}}{Q_{0i}},
\]
hence
\[
p_iP_k(\widehat z_i)=Q_{ki}Q_{0i}.
\]
Substituting this identity into \eqref{eq:partial_matching} gives
\eqref{eq:Q_reconstruction}. In matrix form,
\[
M_{ki}=Q_{ki}Q_{0i},
\qquad
M=Q\,\mathrm{diag}(Q_{01},\ldots,Q_{0N}).
\]
Since \(Q\) is orthogonal and \(Q_{0i}\neq 0\) whenever \(p_i>0\), the matrix
\(M\) is nonsingular, so the reduced reaction-channel levels are uniquely
determined.
\end{proof}

Because \(\{P_k\}_{k=0}^{N-1}\) forms an orthonormal basis of
\(\mathbb P_{N-1}\), the coefficient-matching relation
\eqref{eq:partial_matching} immediately extends to every polynomial
\(q\in\mathbb P_{N-1}\). Writing
\[
q(z)=\sum_{k=0}^{N-1} c_k P_k(z),
\]
and applying \eqref{eq:partial_matching} termwise gives
\[
\int q(z)\,\sigma_x(z)\,d\mu_M(z)
=
\sum_{i=1}^N p_i \sigma_{x,i}\, q(\widehat z_i).
\]
In particular, taking \(q(z)=z^m\) for \(m=0,\dots,N-1\) yields exact matching
of the first \(N\) transformed mixed moments on the realized measure \(\mu_M\).

\begin{remark}
\label{rem:mixed_moment_vs_positivity}
Lemma~\ref{lemma:orthogonal_reconstruction} characterizes the coefficient-matching
relations preserved by the reconstruction, but it does not by itself imply
componentwise nonnegativity of the reconstructed reaction-channel levels
\(\{\sigma_{x,i}\}_{i=1}^{N}\). Accordingly, the present mixed-moment
reconstruction does not furnish a structural guarantee that the resulting
effective cross sections remain in \(\mathbb R_{\ge 0}\) over the target
dilution set \(\mathcal S\).

This distinction reflects the different structures of the two reconstruction
steps. The total subgroup levels and probabilities arise from a
positive-measure compression problem, for which the underlying scalar Gauss
construction naturally preserves reality and nonnegativity. By contrast, the reaction-channel levels are recovered afterwards on the
fixed compressed nodes by orthogonal-basis matching under Chiba's affine-order conditions. In general, exact low-order fidelity to
such mixed moments is neither equivalent to nor necessarily compatible with
componentwise nonnegativity of the reconstructed channel levels.

\end{remark}

\subsection{Reorthogonalization}
\label{subsec:finite_precision_reorth}

The preceding subsection describes the Gauss compression and subgroup
reconstruction in exact arithmetic. In finite precision, however, the three-term
Lanczos recurrence does not preserve exact orthogonality indefinitely; the
resulting loss of orthogonality may produce duplicated or spurious Ritz values
and perturb the projected tridiagonal matrix \(J_N\) \cite{MeurantStrakos2006}.
The compressed nodes and weights extracted by the Golub--Welsch procedure are
therefore perturbed as well, so the practical quality of the compressed rule
depends not only on the realized measure \(\mu_M\), but also on how well
orthogonality is maintained during the Lanczos process.
We consider three reorthogonalization strategies:
\begin{itemize}
\item \textbf{No reorthogonalization}: the three-term recurrence is used
without any explicit correction;
\item \textbf{Selective reorthogonalization}: converged Ritz directions are
periodically identified and locked, while loss of semiorthogonality triggers a
safeguard reorthogonalization step;
\item \textbf{Full reorthogonalization}: each new Lanczos vector is explicitly
reorthogonalized against the current basis.
\end{itemize}

These strategies represent the standard trade-off between recurrence cost and
orthogonality control in finite-precision Lanczos implementations.

The selective strategy used here is a practical hybrid implemented in the spirit
of Parlett--Scott selective orthogonalization, augmented by a semiorthogonality
monitor in the spirit of partial reorthogonalization
\cite{ParlettScott1979,Simon1984}. Converged Ritz directions of the current
projection are periodically identified through residual estimates and then
locked, so that subsequent residuals are kept orthogonal to the locked
subspace. In addition, the semiorthogonality indicator
\[
\mu_k=\max_{i\le k}\frac{|q_i^{\mathsf T}r_k|}{\|r_k\|_2}
\]
is monitored, where \(r_k\) denotes the current Lanczos residual. When
\(\mu_k\) exceeds a prescribed tolerance, a full reorthogonalization step
against the current Lanczos basis is triggered as a safeguard. This hybrid
strategy is intended to suppress the dominant orthogonality loss while
retaining most of the efficiency of the unreorthogonalized recurrence. Its
effect on orthogonality preservation is examined in Section~\ref{sec:results}; the
corresponding results are reported in Figure~\ref{fig:orth-off-by-group}.

\section{Results}
\label{sec:results}

Numerical tests were performed using pointwise cross sections reconstructed
from the ENDF/B-VIII.1 evaluation for the tested cases, with the total cross
section \(\sigma_t(E)\) and the corresponding reaction-channel cross section
\(\sigma_x(E)\) taken as inputs for each case. Two multigroup discretizations,
namely a fine-group structure and a coarse-group structure, were considered in
order to assess sensitivity to the underlying energy partition. For each energy
group, the conventional moment--Pad\'e construction and the proposed
Lanczos--Golub--Welsch construction were compared under the same Chiba
affine-order prescription and otherwise identical settings. The observed
difference therefore reflects the behavior of two finite-precision construction
routes for the same prescription. The reference effective cross section was
evaluated under the narrow-resonance approximation on the same pointwise data
over a logarithmic dilution grid \(\mathcal S\) spanning
\(\sigma_0 \in [10^{-1},10^{6}]\) barn. Unless otherwise noted, the detailed
plots and groupwise discussion below focus on the representative
\(^{238}\mathrm{U}\) capture case, while cross-channel comparisons are
summarized separately.

Groupwise accuracy is measured by the 95th-percentile error \(E_g^{0.95}\) and
the worst-case error \(E_g^{\max}\). The discussion below uses
\(E_g^{0.95}\) as the primary groupwise summary of the observable error over
\(\mathcal S\). In addition to accuracy, two finite-precision nonnegative-real
diagnostics are tracked. At the subgroup level, each energy group is classified
according to whether the reconstructed total subgroup levels and probabilities
are nonnegative real, negative real, or complex. At the response level, we
record whether the reconstructed effective cross sections remain in
\(\mathbb R_{\ge 0}\) over \(\mathcal S\).

\subsection{Effects of reorthogonalization}

Figure~\ref{fig:orth-off-by-group} compares the pairwise orthogonality-loss
indicator \(\max_{i\neq j}|q_i^{\mathsf T}q_j|\) for the three
reorthogonalization strategies for the representative case \(N=50\).
Without reorthogonalization, the loss is strongly group dependent and exhibits
several pronounced spikes, indicating substantial degradation of the Lanczos
basis in a subset of groups. Selective reorthogonalization suppresses these
excursions markedly, whereas full reorthogonalization yields the smallest and
most uniform values over the full group range. The observed ordering is thus
\[
\text{full reorthogonalization}
\;>\;
\text{selective reorthogonalization}
\;>\;
\text{no reorthogonalization},
\]
with respect to orthogonality preservation.

At the level of the final effective-cross-section errors, however, the
differences among the three strategies remain small in the present
experiments. Reorthogonalization therefore changes the internal stability
profile much more visibly than the final observable error. 

A clearer separation appears in the runtime behavior shown in
Figure~\ref{fig:reorth_ratio_vs_M}, which plots the binned ratio
\(t_{\mathrm{sel}}/t_{\mathrm{full}}\) against the ENDF tabulation-point count
\(M\). For \(N=9\), the binned timing trend does not show a clear and
persistent regime favoring selective reorthogonalization. Although some
larger-\(M\) bins fall below \(t_{\mathrm{sel}}/t_{\mathrm{full}}=1\), the
reduction is limited and the pattern is less coherent than in the higher-order
cases. Since the final effective-cross-section errors remain close across the
tested strategies, while full reorthogonalization provides the strongest
orthogonality control, \(N=9\) is treated conservatively using full
reorthogonalization.

For \(N=30\) and \(50\), the binned trends are more structured. Below roughly
\(M\approx 5\times 10^3\), selective reorthogonalization does not show a
stable cost advantage. Above this level, the subsequent bins remain below
unity in their central trend, with the interquartile ranges also concentrated
near or below unity. This identifies a persistent runtime advantage of
selective reorthogonalization in the larger-\(M\) regime. Accordingly,
\(M\approx 5000\) serves as an empirical switching point for the higher-order
cases considered here.

The practical strategy adopted here is therefore as follows:
no reorthogonalization is excluded as a default option;
for \(N=9\), full reorthogonalization is used;
for \(N=30\) and \(50\), full reorthogonalization is retained for
\(M<5000\), whereas selective reorthogonalization is preferred for
\(M\ge 5000\).
This threshold is empirical for the present implementation and dataset.

\begin{figure}[tbp]
  \centering
  \includegraphics[width=0.82\linewidth]{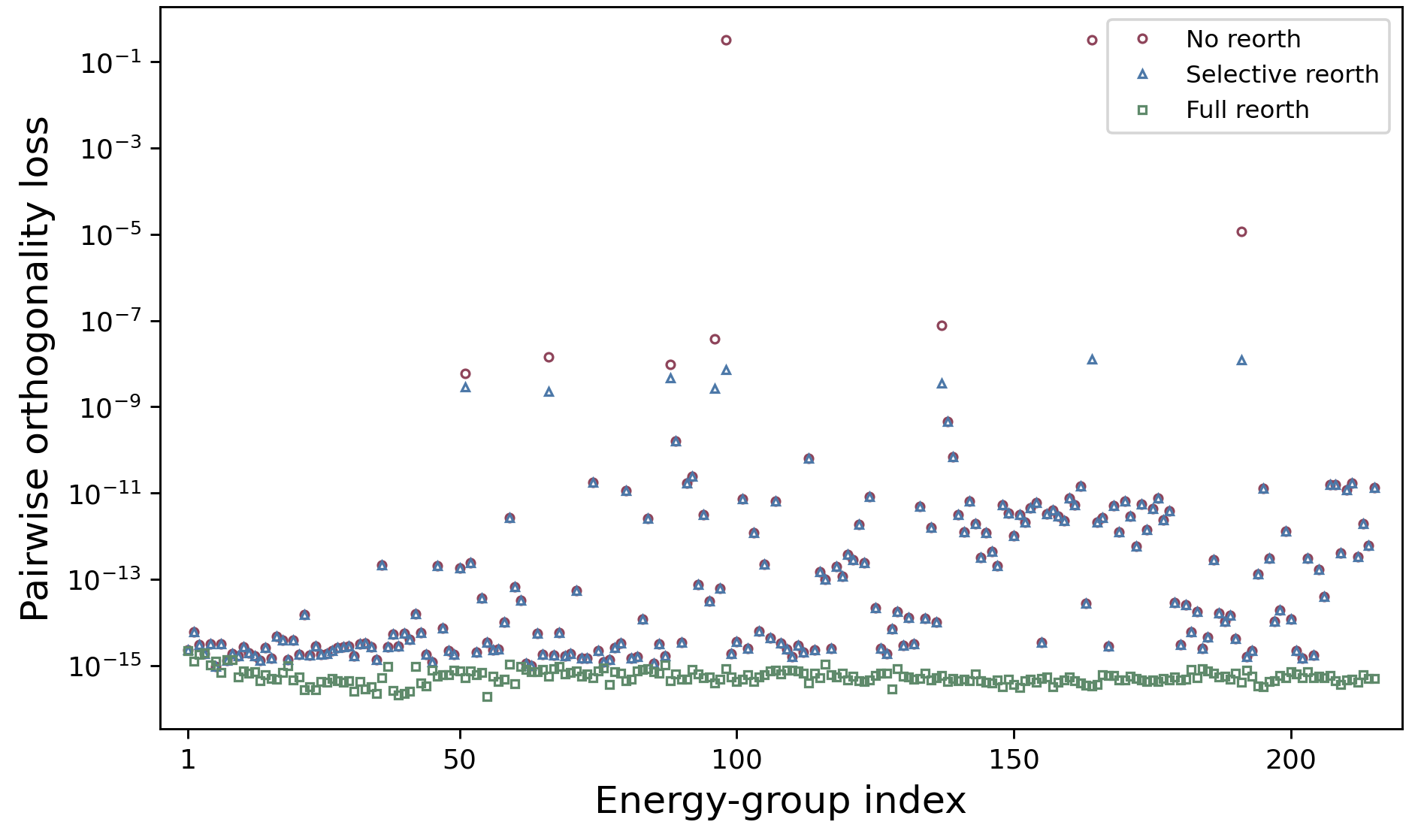}
  \caption{
    Pairwise orthogonality loss of the Lanczos basis versus energy-group index
    for the three reorthogonalization strategies for the representative case
    \(N=50\).
    For each energy group, the metric shown is
    \(\max_{i\neq j}|q_i^{\mathsf T}q_j|\), i.e., the maximum off-diagonal
    inner product among the computed Lanczos vectors.
    The no-reorthogonalization case exhibits several pronounced spikes,
    selective reorthogonalization suppresses these excursions markedly, and
    full reorthogonalization yields the smallest and most uniform values over
    the full group range.
  }
  \label{fig:orth-off-by-group}
\end{figure}

\begin{figure}[tbp]
\centering
\includegraphics[width=0.88\linewidth]{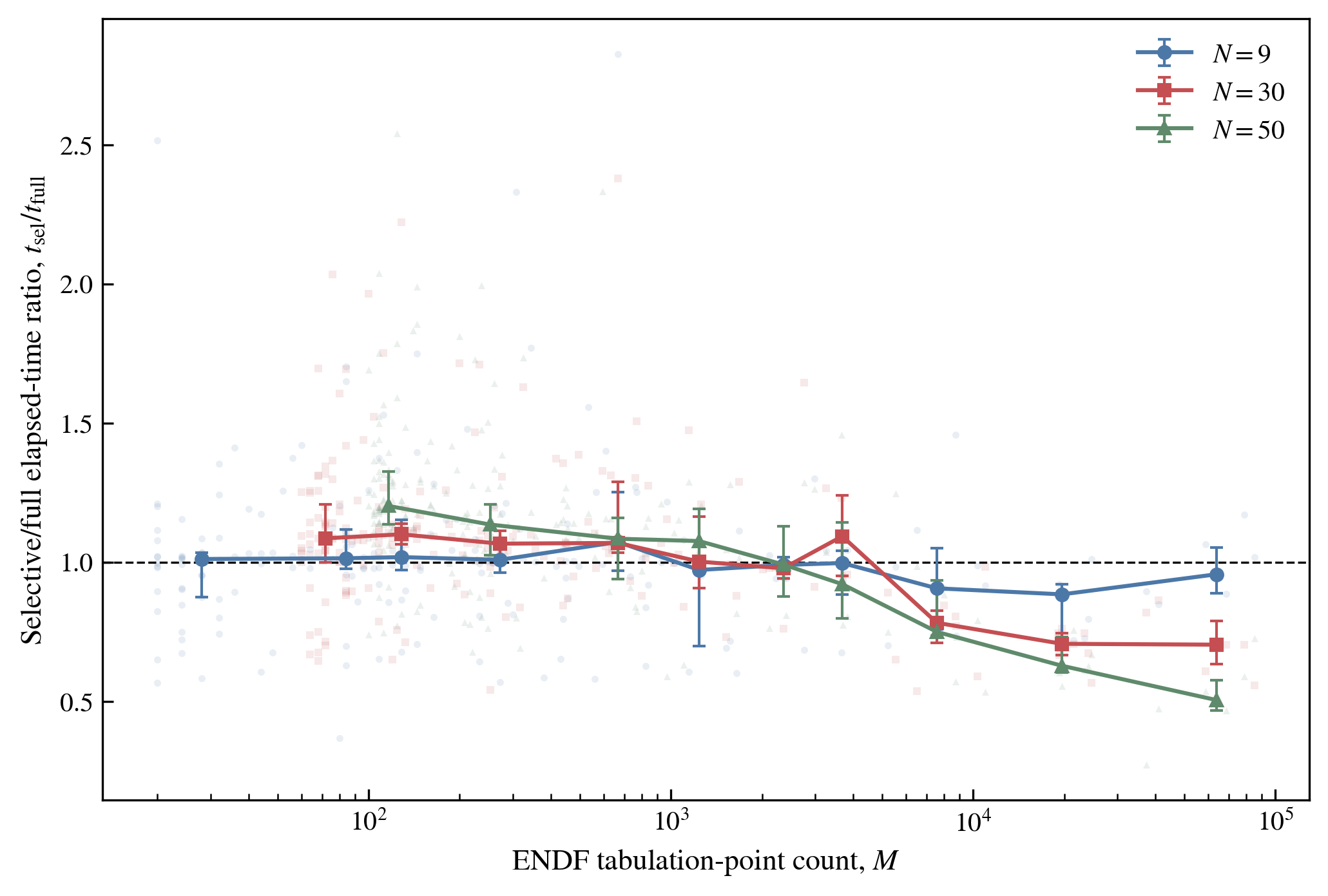}
\caption{
  Size-resolved runtime comparison between selective and full
  reorthogonalization.
  The figure shows the binned ratio
  \(t_{\mathrm{sel}}/t_{\mathrm{full}}\) as a function of the ENDF
  tabulation-point count \(M\) for \(N=9\), \(30\), and \(50\).
  Faint points represent individual energy groups, while the solid markers and
  error bars denote the binwise median and interquartile range.
  For \(N=9\), the timing difference is too small to justify replacing full
  reorthogonalization.
  For \(N=30\) and \(50\), a clearer advantage of selective
  reorthogonalization emerges once \(M\) reaches about \(5000\).
}
\label{fig:reorth_ratio_vs_M}
\end{figure}

\subsection{Accuracy, numerical nonnegative-real behavior, and order scaling}

We now compare the conventional moment--Pad\'e construction with the proposed
Lanczos--Golub--Welsch construction in terms of accuracy, order dependence, and
finite-precision nonnegative-real behavior.

Figure~\ref{fig:analytic-q95-multiN} reports the per-group 95th-percentile
relative error \(E_g^{0.95}\) under the segmentwise analytic reference.
For the three representative orders \(N=9\), \(30\), and \(50\), the proposed
method lies below the conventional one in essentially all groups. This ordering
is stable from moderate to high reconstruction order. At the same time, the
proposed-method errors are concentrated in a relatively narrow band, typically
around \(10^{-13}\) to \(10^{-12}\), and the visual difference between
\(N=30\) and \(N=50\) is limited. This indicates that, once the compression
error has been sufficiently reduced, the remaining observable discrepancy is
increasingly governed by the realization of the reference measure prior to
compression.

\begin{figure}[p]
  \centering
  \captionsetup{font=small}
  \includegraphics[width=0.82\linewidth]{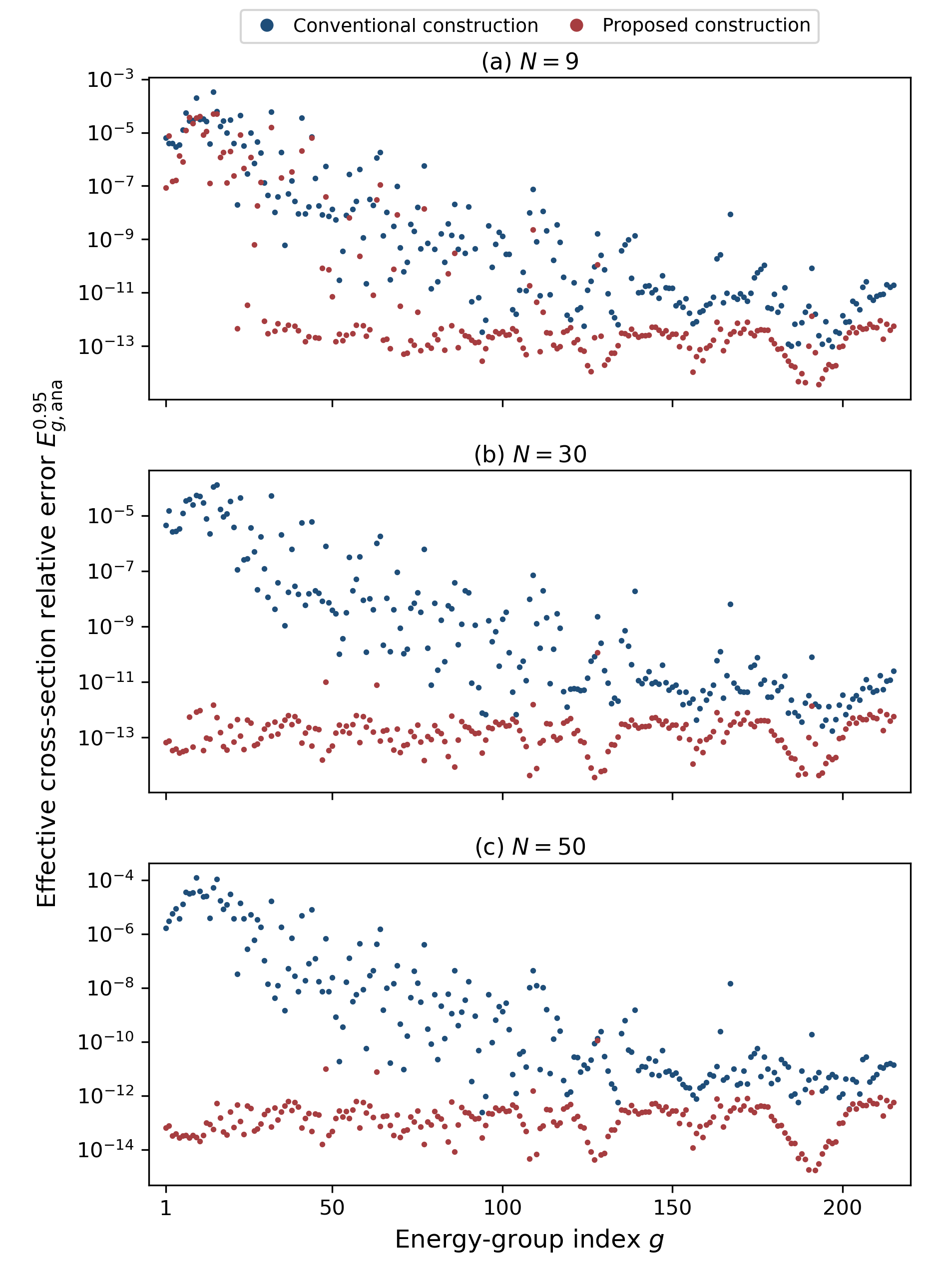}
  \caption[Per-group 95th-percentile relative error under the segmentwise analytic reference]{
    Per-group 95th-percentile relative error of the effective cross section
    for the conventional and proposed constructions, evaluated against the
    segmentwise analytic reference.
    Panels (a)--(c) show the results for \(N=9\), \(N=30\), and \(N=50\),
    respectively.
    Smaller values indicate better agreement with the segmentwise analytic
    reference.
  }
  \label{fig:analytic-q95-multiN}
\end{figure}

This transition is resolved directly in
Figure~\ref{fig:unc-comp-analytic}, which tests the decomposition
\eqref{eq:error_split} at the observable level.
The total discrepancy is separated into a pre-compression realization error and
an additional compression-induced error. At low order, most clearly for
\(N=9\), the compression-induced contribution remains elevated in a subset of
groups and can exceed the pre-compression term. The low-order regime is
therefore still partly compression limited. By \(N=30\) and \(N=50\), however,
the compression-induced contribution is smaller than the pre-compression error
in most groups, typically by one or more orders of magnitude. The moderate- and
high-order regime is therefore realization dominated: once the order is large
enough, the total error under the segmentwise analytic reference is governed
primarily by the baseline realization error, while the additional error
introduced by the Lanczos--Golub--Welsch reduction becomes secondary.

\begin{figure}[p]
  \centering
  \captionsetup{font=small}
  \includegraphics[width=0.82\linewidth]{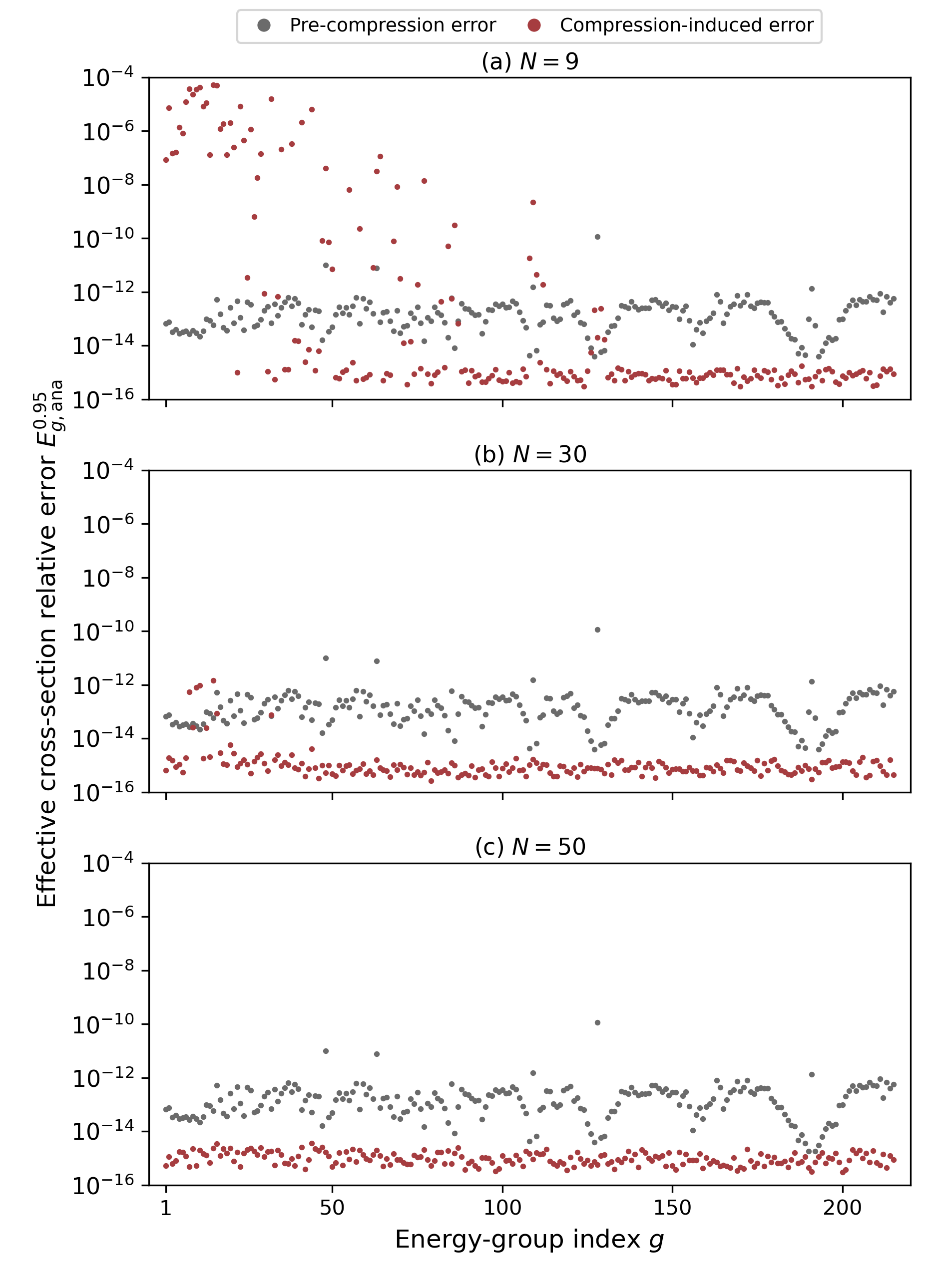}
  \caption{
    Comparison of the uncompressed error and the compression-induced error of
    the proposed construction, evaluated against the segmentwise analytic reference.
    The former reflects the baseline numerical error before compression,
    whereas the latter quantifies the additional error introduced by the
    compression procedure.
    Smaller values indicate better agreement with the segmentwise analytic
    reference.
  }
  \label{fig:unc-comp-analytic}
\end{figure}

Figure~\ref{fig:median-error-ratio-vs-N} summarizes the same regime transition
in a compact order-scaling diagnostic. For each order \(N\), the plotted
quantity is the median over energy groups of the ratio of compression-induced
error to pre-compression realization error. The ratio decreases rapidly with
\(N\), crosses the unity threshold at \(N=26\), and remains below unity for
all tested orders up to \(N=50\). This marks the transition from a
compression-limited regime to a realization-dominated regime under the
segmentwise analytic realization. Beyond the crossover, increasing \(N\) no
longer materially reduces the total observable error, because the dominant term
has already shifted to the pre-compression realization error.

\begin{figure}[htbp]
  \centering
  \includegraphics[width=0.72\linewidth]{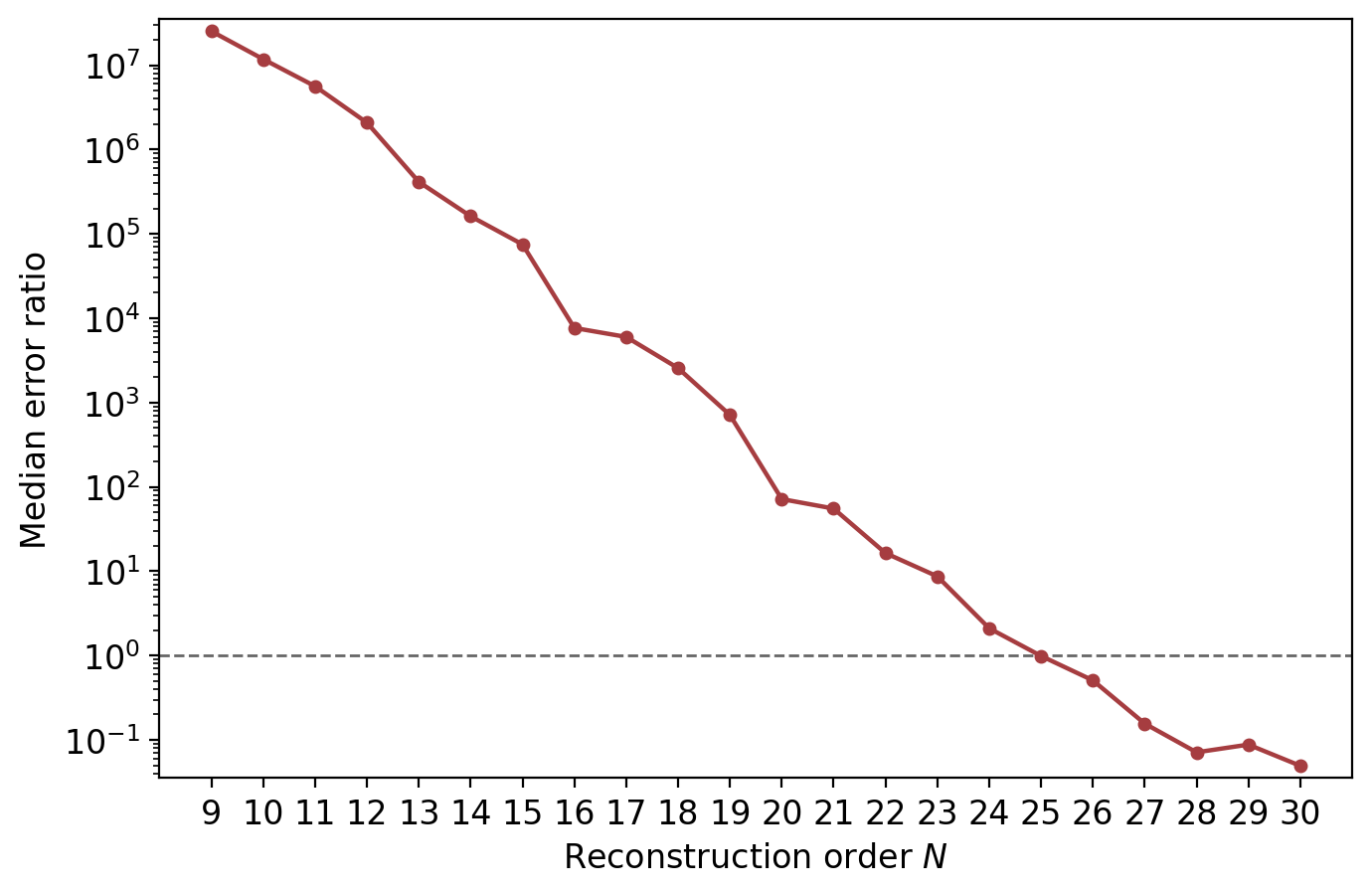}
  \caption{
    Median error ratio versus reconstruction order \(N\).
    For each order \(N\), the plotted value is the median over energy groups
    of the ratio of compression-induced error to pre-compression realization
    error.
    The horizontal line at unity marks the transition between
    compression-dominated and realization-dominated behavior.
    The ratio decreases rapidly with \(N\) and falls below unity at high
    order, indicating a transition from a compression-limited regime to a
    realization-dominated regime.
  }
  \label{fig:median-error-ratio-vs-N}
\end{figure}

Accuracy alone does not capture the main finite-precision contrast between the
two constructions, namely whether the reconstructed subgroup variables remain
nonnegative real.
Figure~\ref{fig:subgroup-nonnegative-real-selective} therefore reports the
groupwise numerical classification of the reconstructed total subgroup levels
and probabilities under the segmentwise analytic realization.
For the proposed Lanczos--Golub--Welsch construction, all groups are classified
as nonnegative real in all six panels, covering both group structures and the
three representative orders \(N=5\), \(N=9\), and \(N=12\).
By contrast, the conventional moment--Pad\'e construction already exhibits
negative-real and complex groups at low order, especially in the fine-group
structure at \(N=5\), and by \(N=9\) shows a clear mixture of nonnegative-real,
negative-real, and complex cases. At \(N=12\), the conventional construction
becomes complex in all displayed groups for both group structures. Thus, at the
subgroup level, the proposed construction remains numerically nonnegative real
throughout the displayed range, whereas the conventional construction undergoes
a rapid order-dependent breakdown.

\begin{figure}[p]
  \centering
  \captionsetup{font=small}
  \includegraphics[width=0.95\linewidth]{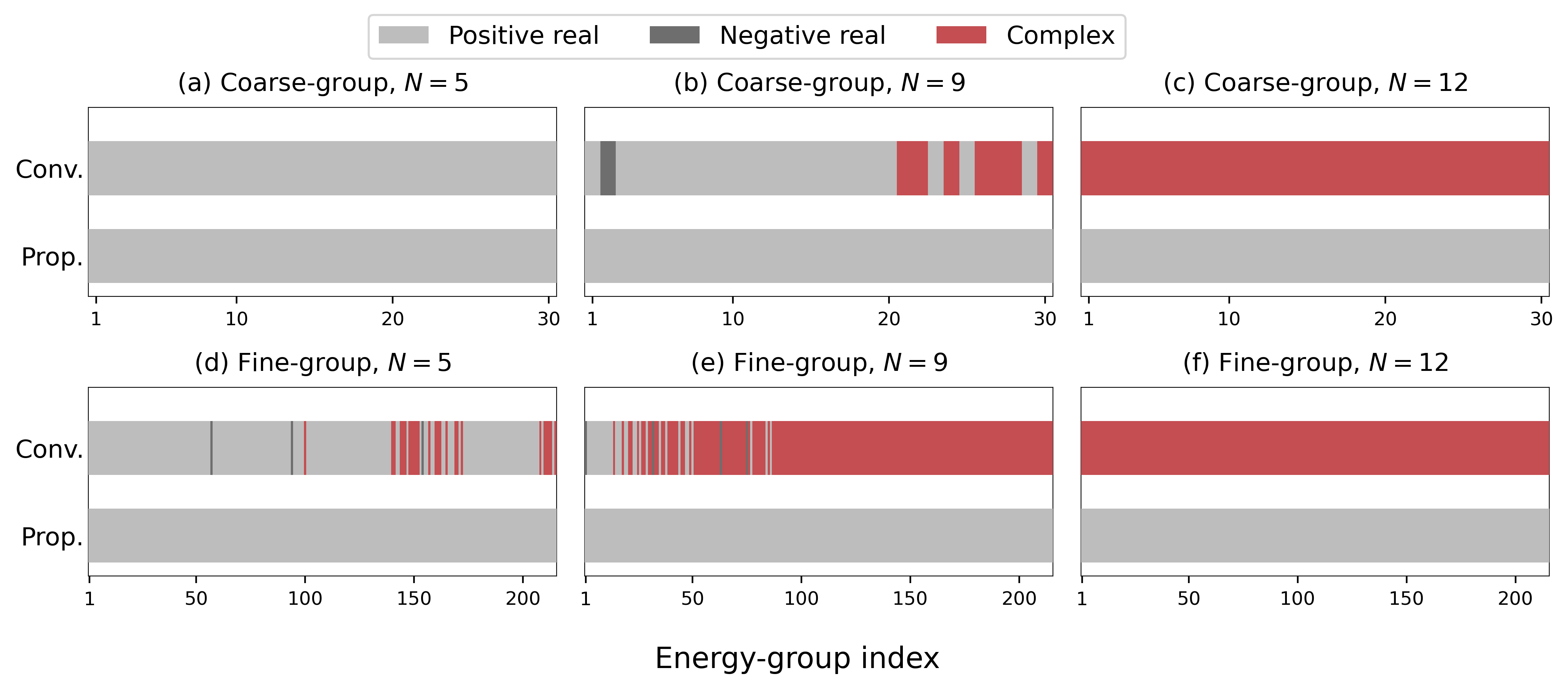}
  \caption{
    Group-wise numerical classification of the reconstructed total subgroup
    levels and probabilities for the conventional and proposed constructions
    under the segmentwise analytic realization.
    The top row corresponds to the coarse-group structure and the bottom row to
    the fine-group structure; panels (a)--(c) and (d)--(f) show the results
    for \(N=5\), \(N=9\), and \(N=12\), respectively.
    Each energy group is assigned to one of three mutually exclusive
    categories: nonnegative real, complex, and negative real, according to the
    reconstructed total subgroup levels and probabilities in that group.
    The proposed construction remains nonnegative real in all displayed groups,
    whereas the conventional construction exhibits a rapid order-dependent
    breakdown and is already complex in all groups by \(N=12\).
  }
  \label{fig:subgroup-nonnegative-real-selective}
\end{figure}

A complementary response-level summary is shown in
Figure~\ref{fig:xg-violation}, which reports the fraction of energy groups
whose reconstructed effective cross sections violate nonnegative-realness over
the target dilution set \(\mathcal S\), for both the conventional
moment--Pad\'e construction and the proposed construction under trapezoidal
and analytic realizations.
For the proposed construction, this violation fraction remains identically
zero over the entire tested order range up to \(N=50\).
By contrast, for the conventional construction, the violation fraction
increases rapidly with reconstruction order, reaches unity at \(N=12\), and
remains equal to unity for all higher tested orders.
The analytic realization reduces the violation fraction relative to the
trapezoidal realization before \(N=12\), but the same complete breakdown still
appears at \(N=12\).
Thus, for the conventional pipeline, increasing the reconstruction order does
not produce regular high-order improvement at the response level; instead, the
high-order regime is dominated by complete nonnegative-real breakdown of the
effective cross sections.

\begin{figure}[htbp]
  \centering
  \includegraphics[width=0.8\linewidth]{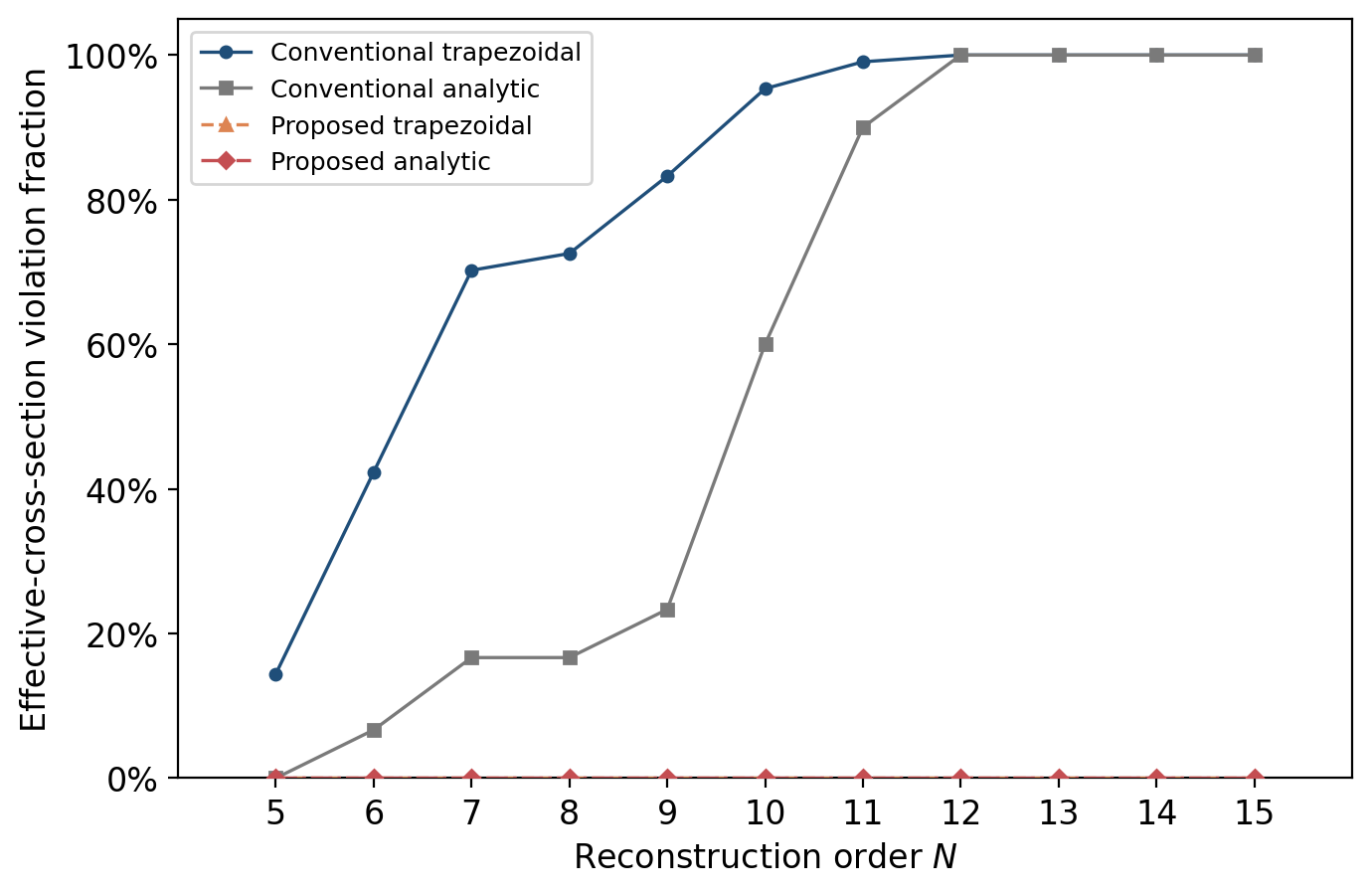}
  \caption[Effective-cross-section nonnegative-real violation fraction versus reconstruction order]{
    Fraction of energy groups for which the reconstructed effective cross
    sections violate nonnegative-realness over the target dilution set,
    plotted against the reconstruction order \(N\).
    Results are shown for both the conventional construction and the proposed construction, each under trapezoidal and analytic
    realizations.
    For the proposed construction, the violation fraction remains identically
    zero over the tested order range, whereas for the conventional
    construction it increases rapidly with \(N\) and reaches unity at high
    order.
  }
  \label{fig:xg-violation}
\end{figure}

Taken together with Figure~\ref{fig:median-error-ratio-vs-N}, these results
show a clear change in the role of subgroup order for the proposed
construction.
At low order, the effective-cross-section error still contains a visible
compression-induced contribution, so increasing \(N\) improves the observable
mainly through better compression.
Once the order is sufficiently high, however, the compression-induced term
becomes subdominant and the remaining observable error is governed primarily by
the pre-compression realization error.
In this higher-order regime, the principal gain from the proposed construction
is numerical robustness rather than continued reduction of the total response
error: it reaches the realization-limited accuracy level while the computed
effective cross sections remain numerically nonnegative real throughout the
tested range up to \(N=50\).
By contrast, the conventional moment--Pad\'e construction exhibits
order-induced numerical breakdown both at the subgroup-variable level and at
the level of the reconstructed effective responses. The present result therefore
establishes a route-level finite-precision advantage rather than a full
stability theorem.

\subsection{Cross-channel tests}
\label{subsec:five_case_summary}

\begin{figure}[htbp]
  \centering
  \includegraphics[width=0.92\linewidth]{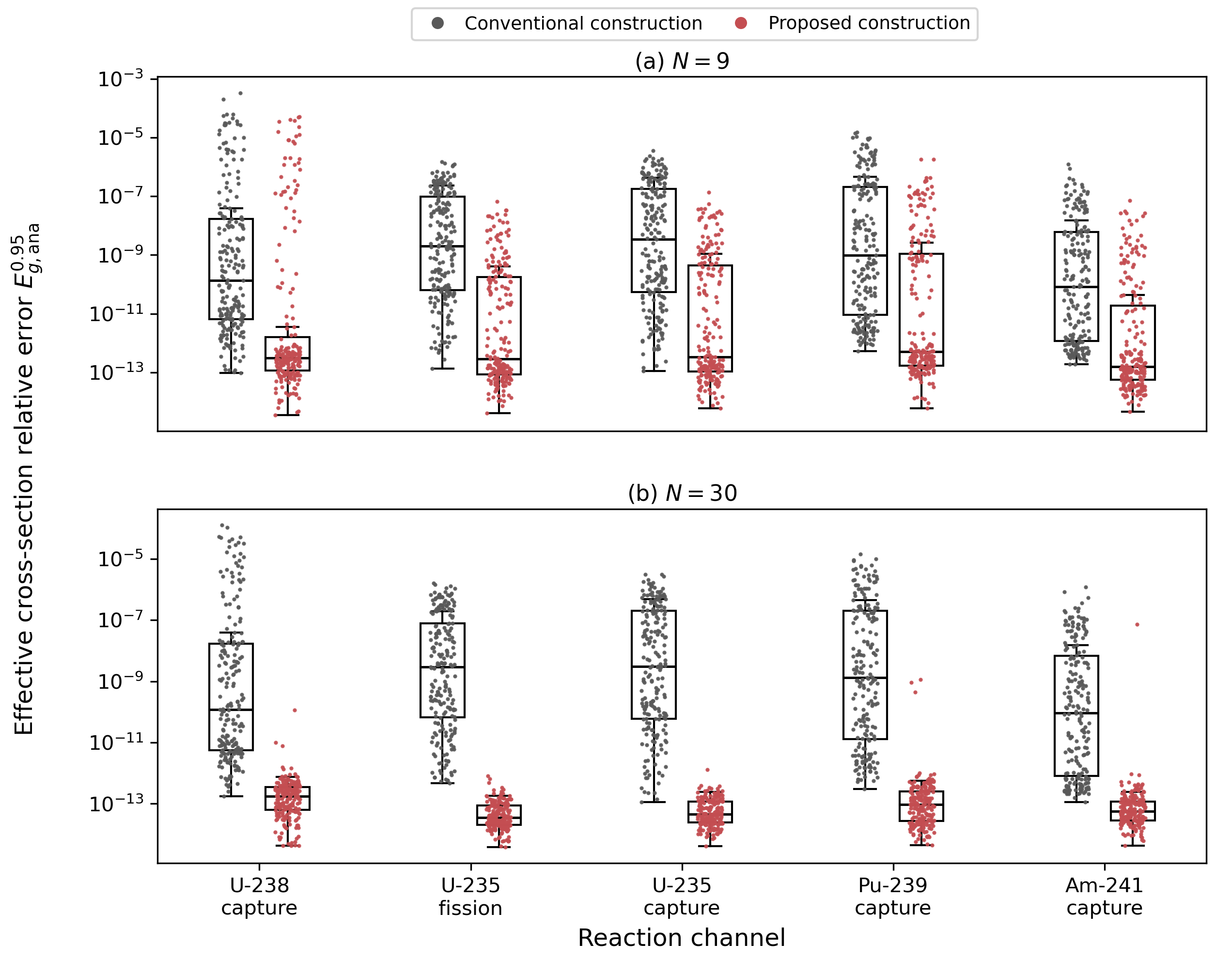}
  \caption{
    Groupwise distributions of effective-cross-section relative errors for five
    tested nuclide--reaction cases at subgroup orders \(N=9\) and \(N=30\):
    \(^{238}\mathrm{U}\) capture, \(^{235}\mathrm{U}\) fission,
    \(^{235}\mathrm{U}\) capture, \(^{239}\mathrm{Pu}\) capture, and
    \(^{241}\mathrm{Am}\) capture. In each pair, the left boxplot denotes the
    conventional construction and the right boxplot denotes the proposed
    construction. The same qualitative finite-precision contrast observed in the
    detailed \(^{238}\mathrm{U}\) capture study is also seen across these five
    tested channels: the proposed construction consistently produces markedly
    smaller error distributions than the conventional route.
  }
  \label{fig:five-cases-error}
\end{figure}

To examine whether the finite-precision contrast identified in the detailed
\(^{238}\mathrm{U}\) capture study reflects a broader cross-channel pattern, we
further considered four additional resonance-channel cases:
\(^{235}\mathrm{U}\) fission, \(^{235}\mathrm{U}\) capture,
\(^{239}\mathrm{Pu}\) capture, and \(^{241}\mathrm{Am}\) capture.
Figures~\ref{fig:five-cases-error} and
\ref{fig:five-cases-response-character} summarize the groupwise response
behavior for the resulting five-case benchmark at \(N=9\) and \(N=30\).

These five-case results show that the qualitative advantage observed for the
representative \(^{238}\mathrm{U}\) capture case persists across the full
benchmark. The proposed construction yields markedly smaller
effective-cross-section error distributions, and nonnegative-real effective
responses are obtained in all five tested channels at both displayed orders.

\begin{figure}[htbp]
  \centering
  \includegraphics[width=0.92\linewidth]{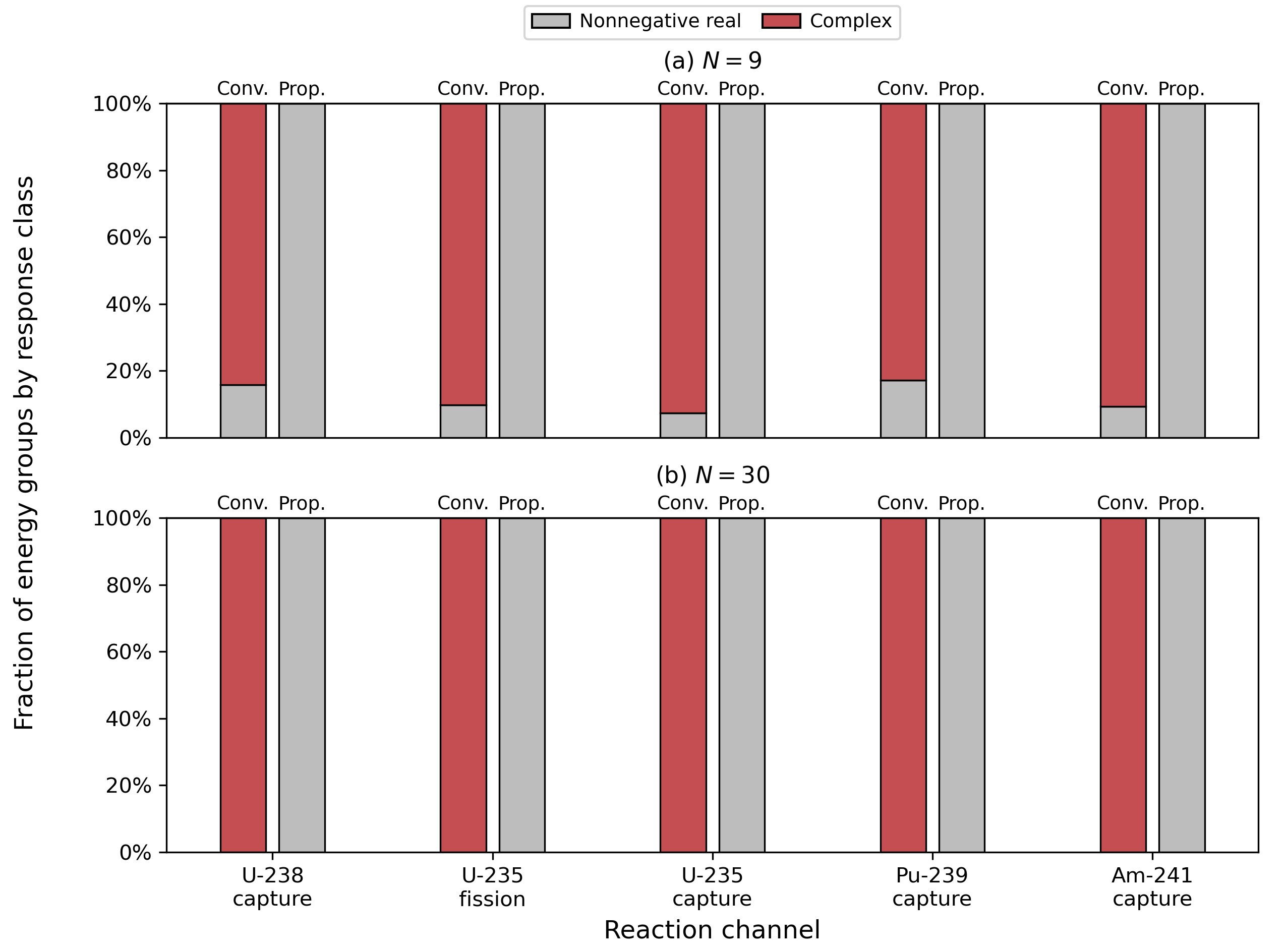}
  \caption{
    Cross-channel summary of the character of the computed effective responses
    for five tested nuclide--reaction cases at subgroup orders \(N=9\) and
    \(N=30\): \(^{238}\mathrm{U}\) capture, \(^{235}\mathrm{U}\) fission,
    \(^{235}\mathrm{U}\) capture, \(^{239}\mathrm{Pu}\) capture, and
    \(^{241}\mathrm{Am}\) capture. For each reaction channel, the left stacked
    bar corresponds to the conventional moment--Pad\'e construction and the
    right stacked bar to the proposed construction. Each bar is normalized to
    unity and partitioned according to the character of the computed effective
    cross section: nonnegative real or complex. Across all five tested cases,
    the proposed construction is classified in the nonnegative-real response
    class at both \(N=9\) and \(N=30\), whereas the conventional construction
    shows a substantial fraction of nonphysical responses, especially at the
    higher order. No real-negative effective responses are observed in these
    five-case summary tests.
  }
  \label{fig:five-cases-response-character}
\end{figure}

\section{Conclusion}
\label{sec5}

This work reformulated probability-table construction in resonance
self-shielding as the compression of a positive measure in transformed cross-section space, under which the physically motivated Chiba affine moments are represented as polynomial moments of a positive transformed measure. On this basis, a new construction route was developed for this
problem, based on discrete measure realization, symmetric Lanczos reduction, and Golub--Welsch extraction, in place of the conventional
moment--Pad\'e pipeline. From the finite-precision viewpoint, the difference between the two
constructions is route-level: the conventional construction proceeds through
explicit moment inversion, Pad\'e root--residue recovery, and a
Vandermonde-type mixed-moment solve, whereas the proposed construction
proceeds through positive-measure realization, symmetric tridiagonal spectral
extraction, and orthogonal-basis channel reconstruction, while retaining the
Gauss-quadrature structure of the underlying approximation problem.
Although this does not constitute a full backward-stability theorem, nor extend
the structure-preserving property beyond the positive-measure compression step,
it does establish a distinct finite-precision construction route for
probability-table approximation in the present setting.

Numerical results show that the conclusions drawn from the representative
$^{238}\mathrm{U}$ capture case extend consistently to the broader five-case
benchmark: the proposed construction improves effective-cross-section accuracy
over the practically relevant low- and moderate-order range and retains a
substantial robustness advantage once the conventional moment--Pad\'e pipeline
enters its order-induced breakdown regime. The error-decomposition study
further shows that the role of compression is realization dependent. Under the
segmentwise analytic realization, the compression-induced contribution is still
significant at low order, but decreases rapidly with increasing order and
eventually becomes secondary to the pre-compression realization error. Thus,
under the analytic realization, the high-order regime should be interpreted not
as one of indefinite further response-error reduction, but as one in which the
compression error has been driven below the realization-imposed floor. Under
the trapezoidal realization, by contrast, the observable error remains
predominantly compression induced over the tested order range.

A further main conclusion is that the proposed construction exhibits a twofold numerical advantage over the conventional moment--Pad\'e construction. First, for the positive-measure compression step, exact arithmetic guarantees
real nonnegative compressed nodes and weights, and hence real nonnegative
mapped total subgroup levels together with nonnegative subgroup probabilities.
In the present finite-precision tests, these total subgroup quantities also
remained numerically nonnegative real throughout the tested order range,
whereas the conventional route could lose this property under reconstruction. Second, although the reaction-channel reconstruction does not furnish a
structural nonnegativity guarantee, the computed effective cross sections were
numerically nonnegative real in the tested cases; by contrast, the conventional
construction exhibited nonphysical response behavior once finite-precision
breakdown set in.

Overall, the present work reformulates probability-table construction in
resonance self-shielding as a transformed positive-measure compression problem and,
on this basis, establishes an alternative construction route
that is shown numerically to be substantially more robust than the conventional moment--Pad\'e pipeline.

\bibliographystyle{elsarticle-num}
\bibliography{GQC_refs}

@article{LiZhangLiuZhangHaoWangJiangLiu2023FineMeshSubgroup,
  author  = {Li, Song and Zhang, Qian and Liu, Lei and Zhang, Yongfa and Hao, Jianli and Wang, Xiaolong and Jiang, Lizhi and Liu, Xiaoya},
  title   = {Analysis of the fine-mesh subgroup method and its feasible improvement},
  journal = {Frontiers in Energy Research},
  year    = {2023},
  volume  = {10},
  pages   = {1036063},
}

@article{Gautschi1982GeneratingOrthogonalPolynomials,
  author  = {Gautschi, Walter},
  title   = {On Generating Orthogonal Polynomials},
  journal = {SIAM Journal on Scientific and Statistical Computing},
  year    = {1982},
  volume  = {3},
  number  = {3},
  pages   = {289--317},
}

@article{AptekarevBuslaevMartinezFinkelshteinSuetin2011Pade,
  author  = {Aptekarev, Alexander I. and Buslaev, Viktor I. and Mart{\'i}nez-Finkelshtein, Andrei and Suetin, Sergey P.},
  title   = {Pad{\'e} approximants, continued fractions, and orthogonal polynomials},
  journal = {Russian Mathematical Surveys},
  year    = {2011},
  volume  = {66},
  number  = {6},
  pages   = {1049--1131},
}

@inproceedings{RosierMaoZmijarevicLealSanchez2022PhysicalPT,
  author    = {Rosier, Emeline and Mao, Li and Zmijarevic, Igor and Leal, Luiz and Sanchez, Richard},
  title     = {Comparison of Two Subgroup Methods Based on the Physical Probability Tables},
  booktitle = {Proceedings of the International Conference on Physics of Reactors ({PHYSOR} 2022)},
  year      = {2022},
  pages     = {2561--2569},
  publisher = {American Nuclear Society},
}

@article{YinZhangLiuHeZhang2025VITAS,
  author  = {Yin, Han and Zhang, Qian and Liu, Xiaojing and He, Hui and Zhang, Tengfei},
  title   = {Study on the subgroup method for the resonance self-shielding calculations in {VITAS}},
  journal = {Nuclear Engineering and Design},
  year    = {2025},
  volume  = {432},
  pages   = {113801},
}

@article{LiZhangZhangZhao2020FineGroupSlowingDown,
  author  = {Li, Song and Zhang, Zhijian and Zhang, Qian and Zhao, Qiang},
  title   = {Improvements of subgroup method based on fine group slowing-down calculation for resonance self-shielding treatment},
  journal = {Annals of Nuclear Energy},
  year    = {2020},
  volume  = {136},
  pages   = {106992},
}

@article{LiuHeWenZuCaoWu2018PRNSMNECPX,
  author  = {Liu, Zhouyu and He, Qingming and Wen, Xingjian and Zu, Tiejun and Cao, Liangzhi and Wu, Hongchun},
  title   = {Improvement and optimization of the pseudo-resonant-nuclide subgroup method in {NECP-X}},
  journal = {Progress in Nuclear Energy},
  year    = {2018},
  volume  = {103},
  pages   = {60--73},
}

@article{LiuHeZuCaoWuZhang2018PRNSMGlobalLocal,
  author  = {Liu, Zhouyu and He, Qingming and Zu, Tiejun and Cao, Liangzhi and Wu, Hongchun and Zhang, Qian},
  title   = {The pseudo-resonant-nuclide subgroup method based global--local self-shielding calculation scheme},
  journal = {Journal of Nuclear Science and Technology},
  year    = {2018},
  volume  = {55},
  number  = {2},
  pages   = {217--228},
}

@article{Levitt1972ProbabilityTable,
  author  = {Levitt, Leo B.},
  title   = {The Probability Table Method for Treating Unresolved Neutron Resonances in Monte Carlo Calculations},
  journal = {Nuclear Science and Engineering},
  year    = {1972},
  volume  = {49},
  number  = {4},
  pages   = {450--457},
}

@article{StimpsonLiuCollinsClarno2017MPACT,
  author  = {Stimpson, Shane and Liu, Yuxuan and Collins, Benjamin and Clarno, Kevin},
  title   = {A lumped parameter method of characteristics approach and multigroup kernels applied to the subgroup self-shielding calculation in {MPACT}},
  journal = {Nuclear Engineering and Technology},
  year    = {2017},
  volume  = {49},
  number  = {6},
  pages   = {1240--1249},
}

@article{Hebert2009SubgroupProjection,
  author  = {H{\'e}bert, Alain},
  title   = {Development of the Subgroup Projection Method for Resonance Self-Shielding Calculations},
  journal = {Nuclear Science and Engineering},
  year    = {2009},
  volume  = {162},
  number  = {1},
  pages   = {56--75},
}

@article{Hebert2005RibonExtended,
  author  = {H{\'e}bert, Alain},
  title   = {The Ribon Extended Self-Shielding Model},
  journal = {Nuclear Science and Engineering},
  year    = {2005},
  volume  = {151},
  number  = {1},
  pages   = {1--24},
}

@book{BakerGravesMorris1996,
  author    = {Baker, George A. and Graves-Morris, Peter},
  title     = {Pad{\'e} Approximants},
  edition   = {2},
  series    = {Encyclopedia of Mathematics and its Applications},
  volume    = {59},
  publisher = {Cambridge University Press},
  address   = {Cambridge},
  year      = {1996},
  isbn      = {0521450071},
}

@article{Hopkins1982,
  author  = {Hopkins, T. R.},
  title   = {On the sensitivity of the coefficients of Pad{\'e} approximants with respect to their defining power series coefficients},
  journal = {Journal of Computational and Applied Mathematics},
  volume  = {8},
  number  = {2},
  pages   = {105--109},
  year    = {1982},
}

@article{GonnetGuettelTrefethen2013,
  author  = {Gonnet, Pedro and G{\"u}ttel, Stefan and Trefethen, Lloyd N.},
  title   = {Robust Pad{\'e} Approximation via {SVD}},
  journal = {SIAM Review},
  year    = {2013},
  volume  = {55},
  number  = {1},
  pages   = {101--117},
}

@article{Bai2002,
  author  = {Bai, Zhaojun},
  title   = {Krylov subspace techniques for reduced-order modeling of large-scale dynamical systems},
  journal = {Applied Numerical Mathematics},
  volume  = {43},
  number  = {1--2},
  pages   = {9--44},
  year    = {2002},
}

@article{GanderKarp2001,
  author  = {Gander, Martin J. and Karp, Alan H.},
  title   = {Stable computation of high order {G}auss quadrature rules using discretization for measures in radiation transfer},
  journal = {Journal of Quantitative Spectroscopy and Radiative Transfer},
  volume  = {68},
  number  = {2},
  pages   = {213--223},
  year    = {2001},
}

@article{Allen1974PadeGauss,
  author  = {Allen, G. D. and Chui, C. K. and Madych, W. R. and Narcowich, F. J. and Smith, P. W.},
  title   = {Pad{\'e} approximation and gaussian quadrature},
  journal = {Bulletin of the Australian Mathematical Society},
  volume  = {11},
  number  = {1},
  pages   = {63--69},
  year    = {1974},
}

@article{MeurantStrakos2006,
  author  = {Meurant, G{\'e}rard and Strako{\v s}, Zden{\v e}k},
  title   = {The {L}anczos and conjugate gradient algorithms in finite precision arithmetic},
  journal = {Acta Numerica},
  volume  = {15},
  pages   = {471--542},
  year    = {2006},
}

@article{Simon1984,
  author  = {Simon, Horst D.},
  title   = {The {L}anczos algorithm with partial reorthogonalization},
  journal = {Mathematics of Computation},
  volume  = {42},
  number  = {165},
  pages   = {115--142},
  year    = {1984},
}

@article{ParlettScott1979,
  author  = {Parlett, B. N. and Scott, D. S.},
  title   = {The {L}anczos algorithm with selective orthogonalization},
  journal = {Mathematics of Computation},
  volume  = {33},
  number  = {145},
  pages   = {217--238},
  year    = {1979},
}

@article{GolubWelsch1969,
  author  = {Golub, Gene H. and Welsch, John H.},
  title   = {Calculation of {G}auss Quadrature Rules},
  journal = {Mathematics of Computation},
  volume  = {23},
  number  = {106},
  pages   = {221--230},
  year    = {1969},
}

@book{GolubMeurant2010,
  author    = {Golub, Gene H. and Meurant, G{\'e}rard},
  title     = {Matrices, Moments and Quadrature with Applications},
  publisher = {Princeton University Press},
  address   = {Princeton, NJ},
  year      = {2010},
}

@book{Meurant2006LanczosCG,
  author    = {Meurant, G{\'e}rard},
  title     = {The {L}anczos and Conjugate Gradient Algorithms: From Theory to Finite Precision Computations},
  publisher = {Society for Industrial and Applied Mathematics (SIAM)},
  address   = {Philadelphia, PA},
  year      = {2006},
}

@book{GolubVanLoan2013,
  author    = {Golub, Gene H. and Van Loan, Charles F.},
  title     = {Matrix Computations},
  edition   = {4},
  publisher = {Johns Hopkins University Press},
  address   = {Baltimore, MD},
  year      = {2013},
}

@article{ChibaUnesaki2006,
  author  = {Chiba, Go and Unesaki, Hiroshi},
  title   = {Improvement of moment-based probability table for resonance self-shielding calculation},
  journal = {Annals of Nuclear Energy},
  volume  = {33},
  number  = {13},
  pages   = {1141--1146},
  year    = {2006},
}

@inproceedings{RibonMaillard1986,
  author    = {Ribon, P. and Maillard, J. M.},
  title     = {Probability Tables and {Gauss} Quadrature: Application to Neutron Cross-Sections in the Unresolved Energy Range},
  booktitle = {Meeting on Advances in Reactor Physics and Safety},
  address   = {Saratoga Springs, NY, USA},
  month     = sep,
  year      = {1986},
  note      = {Communication presented at the meeting; report number CEA-CONF-8720},
}
	
\end{document}